\documentclass[12pt,final]{article}
\usepackage{amssymb,amsmath,amsthm}
\usepackage{epsfig}
\usepackage{graphics}

\newtheorem{theorem}{Theorem}[section]
\newtheorem{lemma}[theorem]{Lemma}

\newtheorem{claim}[theorem]{Claim}
\newtheorem{definition}[theorem]{Definition}

\newtheorem{example}[theorem]{Example}

 \DeclareMathOperator{\sinc}{sinc} \DeclareMathOperator{\SNR}{\sf SNR}
\DeclareMathOperator{\modd}{  mod} \DeclareMathOperator{\PSD}{\sf PSD} \DeclareMathOperator{\IR}{\sf IR}

\begin{document}
\title{Interference Alignment for Line-of-Sight Channels}
\author{Leonard Grokop, \quad David N. C. Tse, \quad Roy D. Yates}
\maketitle

\begin{abstract}
The fully connected $K$-user interference channel is studied in a multipath environment with bandwidth $W$. We show that when each link consists of $D$ physical paths, the total spectral efficiency can grow {\it linearly} with $K$. This result holds not merely in the limit of large transmit power $P$, but for any fixed $P$, and is therefore a stronger characterization than degrees of freedom. It is achieved via a form of interference alignment in the time domain. A caveat of this result is that $W$ must grow with $K$, a phenomenon we refer to as {\it bandwidth scaling}. Our insight comes from examining channels with single path links ($D=1$), which we refer to as line-of-sight (LOS) links. For such channels we build a time-indexed interference graph and associate the communication problem with finding its maximal independent set. This graph has a stationarity property that we exploit to solve the problem efficiently via dynamic programming. Additionally, the interference graph enables us to demonstrate the necessity of bandwidth scaling for any scheme operating over LOS interference channels. Bandwidth scaling is then shown to also be a necessary ingredient for interference alignment in the $K$-user interference channel.
\end{abstract}

\section{Introduction}
The problem of communicating efficiently in wireless adhoc networks has received much attention of late, the focus
being on how to deal with interference in a shared medium. Traditional approaches based on orthogonalizing users (eg.
TDMA or FDMA) or reusing the spectrum (eg. CDMA, certain modes of 802.11) suffer from poor spectral efficiency.
In particular as the number of users in the system grows, the spectral efficiency of each link vanishes. More recent
approaches include using multi-hop \cite{GK}, distributed MIMO \cite{Tse}, \cite{Fran}, and interference alignment
\cite{Mohammed}, \cite{Jafar1}, \cite{Jafar2}. This work belongs to the latter category. Interference alignment is a
technique that uses appropriate precoding to compact interfering signals into small dimensional subspaces at each
receiver. At the same time, the subspace occupied by the data remains linearly independent of the interference. It was
first applied to a multiple MIMO base station problem in \cite{Mohammed} and shown to be capable of achieving
multiplexing gains distinctly greater than those achievable using conventional signalling schemes. The technique was
then extended in \cite{Jafar2} to show that there were exactly $4M/3$ degrees of freedom in the MIMO X channel (two
MIMO transmitters each desiring to send data to two MIMO receivers) with $M>1$ antennas at each transceiver. Following
this, a more sophisticated interference alignment technique was developed in \cite{Jafar1} for the $K$-user
interference channel with an infinite number of independently faded sub-channels, and used to demonstrate that contrary
to conventional wisdom, the total degrees of freedom of the channel is $K/2$.

Whilst the last approach demonstrates the potential benefits that interference alignment techniques can provide, it
suffers from a number of limitations. Perhaps the foremost is that whilst a degrees of freedom characterization is
useful in the high $\SNR$ limit, it may not be meaningful at moderate $\SNR$'s. This stems from the fact that degrees
of freedom characterizes the asymptotic slope of the spectral efficiency curve and not its actual value. In particular,
it is unclear whether at any fixed $\SNR$ the total spectral efficiency of the system is increasing in proportion to
$K$, or increasing at all. The point here is that \cite{Jafar1} does not contain a scaling law result the likes of
\cite{GK} and \cite{Tse}, that is, it does not tell us what happens to the system capacity as more users enter the
fray. We address this question by constructing a communication scheme that achieves a scaling of system capacity
arbitrarily close to linear.

There is another key limitation. It is natural to interpret the parallel channels used in the interference alignment
scheme of \cite{Jafar1} as corresponding to sub-channels in the frequency domain. This is due to various difficulties
associated with realizing the scheme over independently faded parallel channels in time, most notably the very rapid
and accurate channel measurement that must take place, and the substantial delay incurred. But in order for a large
number of frequency channels to undergo independent fading, significant scattering/multipath is required.

In this work we examine the $K$-user interference channel with {\it limited multipath}. We start by assuming each of
the $K^2$ links consists only of a single physical path with complex gain $h_{ij}$ and delay $\tau_{ij}$ seconds. This
model is a good representation for a line-of-sight (LOS) channel. Following this we generalize to the case where each
link consists of $D$ physical paths. We illustrate a simple and elegant representation of interference alignment in the
time domain, in terms of aligning symbols on a time-indexed interference graph. This interference graph proves to be an
extremely useful tool for both conceptualizing and solving various problems relating to interference in LOS channels.
We identify the problem of communicating on the LOS interference channel with the problem of finding a maximal
independent set in the interference graph, and show how this problem can be solved efficiently using dynamic
programming principles. The simplicity of this approach makes it quite versatile and potentially capable of being
extended to tackle a variety of related problems.

For the remainder of this work power spectral density ($\PSD$) is used in place of signal-to-noise ratio, as we will
later wish to compare schemes that use different bandwidths. Depending on the link delays, it may be possible to
achieve a spectral efficiency as high as $\frac{1}{2}\log_2(1+|h_{ii}|^2{\PSD}/N_0)$ bps/Hz for each user $i$,
regardless of the number of interferers. This is exactly half the spectral efficiency achievable in the absence of
interference. We characterize the precise channel conditions for which this is possible and show that they occur in at
least 1/16 of all scenarios for the $3$-user interference channel.

This of course, says little about the {\it typical} gains one can expect by aligning interference in the time domain. To
address this question, we treat the link delays as independent and uniformly distributed random variables in $[0,T_d)$,
where $T_d$ is the delay-spread of the channel. Previous work such as \cite{Mohammed}, \cite{Jafar1}, \cite{Jafar2} has
focused on degrees of freedom as a metric for performance, which is a measure of the scaling of spectral efficiency
with $\PSD$ when the power spectral density is large. We focus on the {\it scaling of spectral efficiency with $K$}
when the number of users is large. The main result of this work is the construction of a communication scheme that
achieves a spectral efficiency arbitrarily close to $O(1)\log_2(1+|h_{ii}|^2{\PSD}/N_0)$ for each user $i$ as
$K\rightarrow\infty$. Thus our result characterizes the best possible scaling of spectral efficiency with $K$ for any
fixed $\PSD$, as compared with \cite{Mohammed}, \cite{Jafar1}, \cite{Jafar2}, where the best possible scaling of
spectral efficiency with $\PSD$, for any fixed $K$, is characterized. In this sense, our characterization has a similar
flavor to characterizations of scaling laws for wireless adhoc networks such as \cite{GK}, \cite{Tse}, however, our
scheme requires no cooperation between users.

A caveat of our result is that the bandwidth must scale sufficiently with $K$. Interpreting the parallel channel model
of \cite{Jafar1} in the frequency domain, one sees that the bandwidth there must also scale with $K$. Interestingly
enough, the bandwidth scaling required for our scheme is essentially the same as that required in \cite{Jafar1}, namely
$O(K^{2K^2})$. However, whereas the scheme in \cite{Jafar1} requires coding over blocks of length $O(K^{2K^2})$, which
creates significant encoding and decoding complexity issues as well as substantial delay, our scheme requires no block
coding, and consequently does not incur any delay or suffer from complexity issues.

Essentially, in order to align interference into a small dimensional subspace and keep it linearly independent of the
data subspace, a high degree of resolvability of the received signals is required. We establish this concretely for the
LOS channel in the context of the interference graph, showing that if the bandwidth scales sub-linearly with $K$, then
the total spectral efficiency of the system (the sum of all users spectral efficiencies) will scale sublinearly, and
hence almost all of the users will witness vanishing spectral efficiencies as $K$ increases.

This suggests that the greatest performance gains can be reaped in systems with large delay spreads. Perhaps the best
example of such a system is the backbone of a mesh network, used to wirelessly connect rural areas. Such systems are
well approximated by a LOS model, have large delay spreads, and are relatively static, making channel measurement
simpler and more accurate.

The structuring of the rest of the paper is as follows. In section \ref{sec:model} we describe the model of the
$K$-user LOS interference channel. Section \ref{sec:preview} provides a summary of the main result concerning the
achievability of non-vanishing spectral efficiencies as the number of users grows. The time-indexed interference graph
is introduced in section \ref{sec:intgraph}. In the same section we present an algorithm for optimizing the spectral
efficiency efficiently via dynamic programming. We also address the questions of bandwidth scaling, and of
characterizing the class of channels for which the spectral efficiency can reach its maximum value. In section
\ref{sec:non_vanish_se} we present our construction that establishes the main result. Following this, in section
\ref{sec:fdomain}, we establish the relationship between time and frequency domain interference alignment techniques.
Section \ref{sec:conc} contains further discussion, extensions and open problems.

\section{Model}\label{sec:model}

We consider the $K$-user interference channel in which there are $K$ transmitters and $K$ receivers. Transmitter $i$
wishes to send data to receiver $i$ but its transmission constitutes interference at all other receivers. We often
refer to each tx-rx pair as a {\it user}. There are thus $K^2$ links in total, $K$ direct links and $K(K-1)$
cross-links. Each link consists of a single physical path. Denote the gain and delay (in seconds) of the link between
transmitter $j$ and receiver $i$ by $h_{ij} \in \mathbb C$ and $\tau_{ij} \in [0,\infty)$, respectively. We assume the
$h_{ij}$ and $\tau_{ij}$ are fixed for the duration of communication. Denote the signal transmitter $j$ sends by
$x_j(t)$. Then the baseband signal at the $i$th receiver is
\begin{equation*}
y_i(t) = \sum_{j=1}^K h_{ij}x_j(t-\tau_{ij}) + z_i(t)
\end{equation*}

\noindent where the $z_i(t)$ are i.i.d. white noise processes with power spectral density $N_0$ Watts/Hz. Denote the
carrier frequency $f_c$ and the bandwidth that the signals $x_j(t)$ are constrained to lie in by $W$ Hz. Assuming the
use of ideal sinc pulses, the passband signal after sampling is given by
\begin{equation*}
y_i[m] = \sum_{j=1}^K h_{ij}e^{-2\pi f_c \tau_{ij}}\sum_{l=0}^\infty \sinc(l - \tau_{ij}W) x_j[m-l] + z_i[m]
\end{equation*}

\noindent where the $z_i[m]$ are i.i.d. ${\cal CN}(0,N_0W)$. We use the following conventional approximation for the
$\sinc$ pulse,
\begin{equation*}
\sinc(t) \approx \left\{
             \begin{array}{ll}
               1, & \hbox{if $-1/2<t < 1/2$;} \\
               0, & \hbox{otherwise,}
             \end{array}
           \right.
\end{equation*}

\noindent see page 27 of \cite{book}. Let $l_{ij}$ denote the integer round-off of the real number $\tau_{ij}W$. This
leads to
\begin{equation*}
y_i[m] = \sum_{j=1}^K h_{ij}e^{-2\pi f_c \tau_{ij}} x_j[m-l_{ij}] + z_i[m].
\end{equation*}

In a wireless model one typically makes an assumption about the statistics of the channel. For example, in a channel
with ISI, the tap coefficients are often modeled as i.i.d. Rayleigh random variables. In models where there is a
dominant path, Rician random variables are used instead. Likewise in this work we make a statistical assumption on the
channel, but limit this statistical assumption only to the link delays, taking the $\tau_{ij}$ to be i.i.d. uniform in
$[0,T_d)$, where $T_d$ denotes the delay-spread of the channel in seconds. This means that if we define $L$ to be one
plus the integer round off of $T_dW$, the $l_{ij}$ are i.i.d. uniform in $\{0,\dots,L-1\}$. No assumption is made on
the link gains $h_{ij}$, other than they all being non-zero.

We refer to this as the $K$-user {\it line-of-sight (LOS)} interference channel, as this is the most common scenario
giving rise to such a model.

There is a straightforward extension of this model to the case where each link consists of $D$ physical paths, such
that the $i$th received signal is
\begin{equation*}
y_i(t) = \sum_{j=1}^K \sum_{d=1}^D h_{ij,d}x_j(t-\tau_{ij,d}) + z_i(t).
\end{equation*}

\noindent Here $h_{ij,d}\in{\mathbb C}$ and $\tau_{ij,d} \in [0,\infty)$ are the complex gain and delay of the $d$th
physical path between transmitter $j$ and receiver $i$. This leads to the following generalization of the passband
model after sampling
\begin{equation*}
y_i[m] = \sum_{j=1}^K \sum_{d=1}^D h_{ij,d}e^{-2\pi f_c \tau_{ij,d}} x_j[m-l_{ij,d}] + z_i[m].
\end{equation*}

\noindent The natural extension of our statistical assumption for the LOS channel is to treat the delays $l_{ij,d}$ as
i.i.d. uniform in $\{0,\dots,L-1\}$. In doing so we are assuming independent delays not just across physical paths of
different links, but also across the physical paths corresponding to the same link.

We refer to this as the {\it $K$-user $D$-path interference channel}.

\section{Preview of Main Result}\label{sec:preview}

\begin{theorem}\label{thm:mr}
For any $\epsilon>0$, there exists a communication scheme on the $K$-user LOS interference channel such that if
$W>(2K(K-1))^{K(K-1)+\epsilon}$, the expected spectral efficiency of user $i$ tends to
\begin{equation}\label{eqn:Erate}
\frac{1}{(K(K-1))^{\epsilon}} \log_2\left(1+|h_{ii}|^2\frac{{\PSD}}{N_0}\right)
\end{equation}

\noindent as $K \rightarrow \infty$.
\end{theorem}

Section \ref{sec:non_vanish_se} is devoted to proving this result. Here the expectation of spectral efficiency is taken
over the random direct delays $l_{ii}$. Roughly speaking this result says that as $K$ scales, it is possible for each
user to communicate at a spectral efficiency arbitrarily close to $O(1)\log_2(1+\frac{{\PSD}}{N_0})$, so long as the
bandwidth scales as fast as $O((2K(K-1))^{K(K-1)})$. In other words, communication at spectral efficiencies that vanish
arbitrarily slowly with $K$ is possible if the bandwidth scales sufficiently. In this result, and throughout this work, we assume that all cross delays $l_{ij}$ for $i\neq j$ are known at all transmitters and receivers. This result is of a similar nature to the scaling laws of \cite{Tse}, in that we show the growth of system capacity with the number of users is arbitrarily close
to linear. Unfortunately, as is the case in \cite{Jafar1}, the required bandwidth scaling is great. A result presented
later on (theorem \ref{thm:W_scaling}) addresses the question of whether bandwidth scaling is necessary.

For the case where each link consists of $D$ physical paths we have the following generalization.
\begin{theorem}\label{thm:mr2}
For any $\epsilon>0$, there exists a communication scheme on the $K$-user $D$-path interference channel such that if
$W>(2DK(DK-1))^{DK(DK-1)+\epsilon}$, the expected spectral efficiency of user $i$ tends to
\begin{equation}\label{eqn:Erate}
\frac{1}{(DK(DK-1))^{\epsilon}} \log_2\left(1+\max_{d \in \{1,\dots,D\}} |h_{ii,d}|^2\frac{{\PSD}}{N_0}\right)
\end{equation}

\noindent as $K \rightarrow \infty$.
\end{theorem}

As is evident from the statement of the theorem, there is a tradeoff between the number of physical paths per link and
the bandwidth scaling required. Again, the question of whether this bandwidth scaling is necessary is addressed in a
result presented later on.

\section{The Interference Graph}\label{sec:intgraph}

The key insight leading to theorem \ref{thm:mr} comes from formulating the communication problem in a graph theoretical
setting. We start with an example. Consider a $3$-user LOS interference channel where the direct links are all zero,
i.e. $l_{11} = 0$, $l_{22} = 0$ and $l_{33} = 0$, and the cross-links are say, $l_{21}=3$, $l_{31}=1$, $l_{12}=1$,
$l_{32}=4$, $l_{13}=3$, $l_{23}=0$. Choose a length $T$ for the communication block. Create a directed graph ${\cal
G}_{3,T}(l_{12},l_{13},l_{21},l_{23},l_{31},l_{32}) = ({\cal V},{\cal E})$ as follows. Let the vertex set be
\begin{equation*}
{\cal V} = \{v_1(0),\dots,v_1(T-1)\}\cup\{v_2(0),\dots,v_2(T-1)\}\cup\{v_3(0),\dots,v_3(T-1)\}.
\end{equation*}

\noindent The vertex $v_i(t)$ represents the $t$th time slot for the $i$th transmitter. Form the edge set $\cal E$ as
follows. Add a directed edge $e_{21}(0)$ starting from vertex $v_1(0)$ and ending at vertex $v_2(l_{21})=v_2(3)$. This
represents the fact that, owing to a delay of 3 time slots, a transmission during time slot 0 by transmitter 1 arrives
at receiver 2 during time slot 3. Also add a directed edge $e_{31}(0)$ starting from $v_1(0)$ and ending at
$v_3(l_{31}) = v_3(1)$. This represents the fact that, owing to a delay of 1 time slot, a transmission during time slot
0 by transmitter 1 arrives at receiver 3 during time slot 1. Likewise add directed edges $e_{12}(0)$ and $e_{32}(0)$
from vertex $v_2(0)$ to $v_1(l_{12}) = v_1(1)$ and $v_3(l_{32})=v_3(4)$, respectively, and directed edges $e_{13}$ and
$e_{23}$ from vertex $v_3(0)$ to $v_1(l_{13})=v_1(3)$ and $v_2(l_{23})=v_2(0)$, respectively. This set of six edges
encapsulates all of the interference generated by transmissions during time slot 0. As the channel is time-invariant
the same interference structure applies for later time slots. Thus for each $t=1,2,\dots,T-1$ add a directed edge from
$v_1(t)$ to $v_2(t+l_{21})$, provided $t+l_{21}\le T$, a directed edge from $v_1(t)$ to $v_3(t+l_{31})$, provided
$t+l_{31}\le T$, a directed edge from $v_2(t)$ to $v_1(t+l_{12})$, provided $t+l_{12}\le T$, etc... See figure
\ref{fig:int_graph} for an illustration.

\begin{figure}
\centering
\includegraphics[width=350pt]{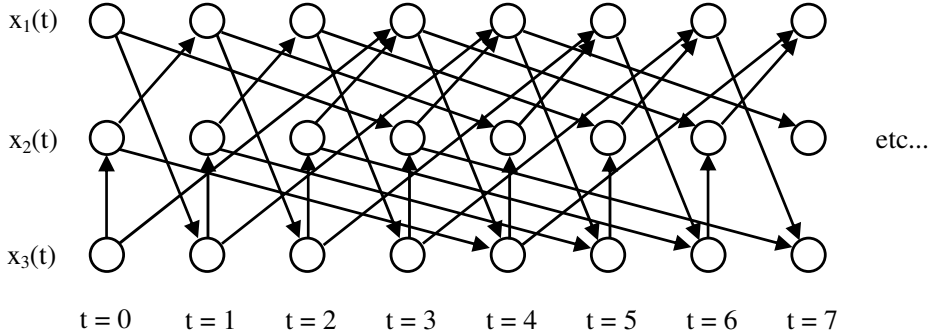}
\caption{The interference graph associated with the LOS channel with direct-delays $l_{11}=0,l_{22}=0,l_{33}=0$ and
cross-delays $l_{21}=3,l_{31}=1,l_{12}=1,l_{32}=4,l_{13}=3,l_{23}=0$.} \label{fig:int_graph}
\end{figure}

In this example all direct delays were zero, whereas in general this is not the case. However as each transmitter $j$
can merely offset it's transmitted sequence $x_j[m]$ by $-l_{jj}$, we can effectively assume without loss of generality
that $l_{jj}=0$. More concretely we define the {\it normalized cross-delays}
\begin{equation*}
l_{ij}' \triangleq l_{ij} - l_{jj}.
\end{equation*}

\noindent Note that $l'_{ij} \in \{-L+1,\dots,L-1\}$, that is, it is possible for $l'_{ij}$ to be negative. At this
point one may wonder why the interference graph need be directed, since the feasibility of a given transmit pattern is independent
of edge direction. The answer is it need not be, but we define it as such to aid in conceptualizing the problem.

\begin{figure}
\centering
\includegraphics[width=350pt]{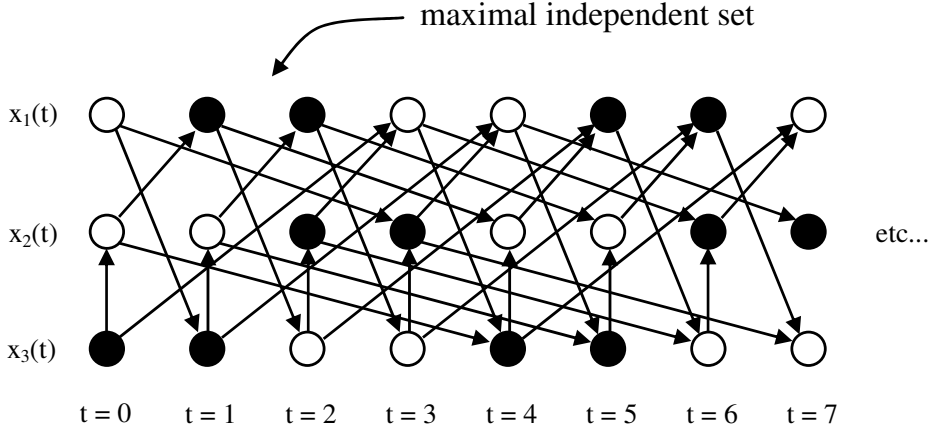}
\caption{A feasible transmit pattern corresponds to an independent set. The vertices of a maximal independent set are
shaded. For a symmetric channel, the maximal independent set maximizes total spectral efficiency.}
\label{fig:max_ind_set}
\end{figure}

In general we have
\begin{definition}
The time-indexed interference graph (or simply interference graph for short) of length $T$ associated with the $K$-user LOS interference channel with normalized cross-delays $\{l'_{ij}\}_{i\neq j}$, is the directed graph ${\cal
G}_{K,T}(\{l'_{ij}\}_{i\neq j}) = ({\cal V},{\cal E})$ where
\begin{align*}
{\cal V} &= \bigcup_{j=1}^K \{v_j(1),\dots,v_j(T) \} \\
{\cal E} &= \bigcup_{i=1}^K\bigcup_{j=1,j\neq i}^K \{e_{ij}(1),\dots,e_{ij}(T-l'_{ij}) \},
\end{align*}

\noindent with edge $e_{ij}(t)$ stemming from vertex $v_j(t)$ and ending at vertex $v_i(t+l'_{ij})$.
\end{definition}

This graph has $KT$ vertices and approximately $K(K-1)T$ edges. As a transmission during time slot $t$ is interfered
with by one time slot at each other user, and itself interferes with one time slot at each other user, each vertex has
both in-degree $K-1$ and out-degree $K-1$. To reduce the notational burden, we often refer to the graph ${\cal
G}_{K,T}(\{l'_{ij}\}_{i\neq j})$ simply as ${\cal G}$, where the parameters of the interference graph are implicit.

One can similarly define a time-indexed interference graph for the $D$-path interference channel, but in the interest of brevity and clarity we do not discuss it here.

%\begin{definition}
%The time-indexed interference graph of length $T$ associated with the $K$-user $D$-path interference channel with
%normalized cross-delays $\{l'_{ij,1}\}_{i\neq j},\dots,\{l'_{ij,D}\}_{i\neq j}$, is the directed graph ${\cal
%G}_{K,T}(\{l'_{ij,d}\}_{i\neq j,d}) = ({\cal V},{\cal E})$ where
%\begin{align*}
%{\cal V} &= \bigcup_{j=1}^K \{v_j(1),\dots,v_j(T) \} \\
%{\cal E} &= \bigcup_{i=1}^K\bigcup_{j=1,j\neq i}^K \bigcup_{d=1}^D \{e_{ij,d}(1),\dots,e_{ij,d}(T-l'_{ij,d}) \},
%\end{align*}
%
%\noindent with edge $e_{ij,d}(t)$ stemming from vertex $v_j(t)$ and ending at vertex $v_i(t+l'_{ij,d})$.
%\end{definition}
%
%\noindent This graph differs from the LOS interference graph in that each vertex has in-degree $D(K-1)$ and out-degree
%$D(K-1)$.

A {\it transmit pattern} is a subset of time slots during which data symbols are sent, one data symbol being sent per
time slot. A transmit pattern is called {\it feasible} if each data symbol is received during a time slot that contains
no interference from other transmissions. Thus a feasible transmit pattern corresponds to an {\it independent set} on
the interference graph. Occasionally we will drop the adjective ``feasible'' when it is clear that the transmit pattern
in question is such. As each data symbol arriving at receiver $i$ during a time slot containing no interfering symbol,
is capable of conveying $\log_2(1+|h_{ii}|^2\frac{\PSD}{N_0})$ bps/Hz, user $i$'s spectral efficiency will be
\begin{equation*}
R_i = N_i \log_2\left(1+|h_{ii}|^2\frac{\PSD}{N_0}\right),
\end{equation*}

\noindent where $N_i$ is the number of vertices in $\{v_i(1),\dots,v_i(T)\}$ that are in the independent set. Let us
assume for the meantime that $h_{ii} = 1$ for all $i$. Then the total spectral efficiency is directly proportional to
the size of the independent set. Thus the problem of designing a communication scheme to maximize total spectral
efficiency reduces to finding the {\it maximal independent set} of the interference graph. Denote the size of the
maximal independent set of a graph ${\cal G}$ (called the {\it independence number}) by $\alpha({\cal G})$. Then the
maximum total spectral efficiency for a graph ${\cal G}$ is
\begin{equation*}
\alpha({\cal G}) \log_2\left(1+\frac{\PSD}{N_0}\right).
\end{equation*}

For the preceding example, the maximal independent set is illustrated in figure \ref{fig:max_ind_set}. Whenever an
independent set contains two vertices that possess a mutual neighbor, the interference generated by these two
transmissions aligns at the mutual neighboring vertex. This is {\it interference alignment in the time domain}. For the
example in figure \ref{fig:max_ind_set}, the neighbors of each unshaded vertex are all shaded, that is, data
transmissions occur at all neighbors of an unshaded vertex. This is not always the case. Suppose we use a TDMA based
communication scheme where user 1 transmits on consecutive time slots for a long period whilst users 2 and 3 remain
silent. Then after a small guard interval user 2 transmits on consecutive time slots whilst users 1 and 3 remain
silent. Finally user 3 transmits, and then back to user 1 and so on. In this round robin scheme, each unshaded vertex
is connected to only a single shaded one and no interference alignment occurs.

The problem of finding the maximal independent set is a well-known NP-hard problem, meaning that for an arbitrary
graph, there is no known algorithm capable of solving the problem in time sub-exponential in the number of vertices.
Knowing this it may appear that finding an optimal transmit pattern requires a computation time that is exponential in
the block length $T$, however the interference graph is not an arbitrary graph. In particular it is stationary in the
sense that, ignoring boundary effects, the structure of the graph is invariant to time shifts. In the next section we
present an algorithm that exploits this property to find the maximum independent set in linear time. More generally,
when the link gains $h_{ii}$ are arbitrary the algorithm solves the problem of finding an independent set that
maximizes spectral efficiency. We refer to this set as the {\it optimal independent set}.

\subsection{Finding the maximal independent set efficiently}

In this section we concentrate on the LOS channel, but the ideas can be extended to the $D$-path channel. Given an
interference graph ${\cal G}$ we now illustrate how dynamic programming principles can be employed to compute the
maximal independent set efficiently. Let each vertex $v_j(t) \in \{0,1\}$ with $v_j(t)=1$ if $v_j(t)$ is included in
the transmit pattern, that is, if a data symbol is transmitted by user $j$ during time slot $t$, and $v_j(t)=0$
otherwise. There is a slight abuse of notation here as we have used $v_j(t)$ to represent both an element of the vertex
set $\cal V$ and an indicator function for whether or not a data symbol is transmitted by transmitter $j$ during time
slot $t$.

\begin{figure}
\centering
\includegraphics[width=350pt]{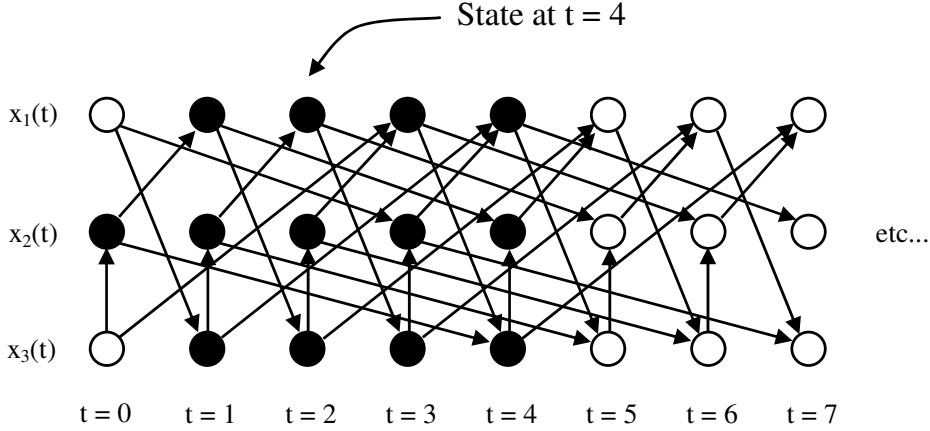}
\caption{The LOS interference channel of figure \ref{fig:int_graph}. The state at time $t=4$ is a function of the
shaded vertices, each taking on one of two values.} \label{fig:state_space}
\end{figure}

Concretely stated, the optimization problem we will solve is
\begin{equation}\label{eqn:op_problem}
\min_{\scriptsize \begin{array}{c}
        \{v_j(t)\}\in\{0,1\}^{KT} \\
        \text{ s.t. } v_i(s) + v_j(t) \le 1, \forall (v_i(s),v_j(t)) \in {\cal E} \\
      \end{array}}
-\sum_{j=1}^Kr_j\left(\sum_{t=0}^{T-1} v_j(t),\right)
\end{equation}

\noindent where $r_i = \log_2(1 + |h_{ii}|^2 \PSD/N_0)$. This is the problem of finding the optimal independent set.
This cost is just the sum of the spectral efficiencies of the users weighted by the number of data symbols they send.
We use a negative sign so as to justify the description of this metric as a {\it cost}, i.e. something we are trying to
minimize. In the event that all direct gains are equal, $r_j$ is independent of $j$ and the problem reduces to finding
the maximal independent set of ${\cal G}$.

To solve this problem efficiently, we start by defining
\begin{equation*}
l_j^* \triangleq \max\{\max_i l'_{ij},\max_i -l'_{ji}\}.
\end{equation*}

\noindent This is the length of the longest edge that connects a vertex at a time $t$, to a vertex belonging to user
$j$ for time $\le t$. When all the $l_{ij}'$ are positive, $l_j^*$ simply represents the longest edge stemming from
user $j$. In the example of figure \ref{fig:int_graph}, $l_1^* = 3$, $l_2^* = 4$ and $l_3^* = 3$. When some of the
$l_{ij}'$ are negative, $l_j^*$ represents the longest edge connecting user $j$ to another user in the forward time
direction. Thus the total amount of memory in the system is $\max_j l_j^*$. If this seems somewhat contrived, recall
that although the interference graph is a directed graph, it need not be defined as such, as the effect of vertex
$v_i(t)$ causing interference at vertex $v_j(t')$ is identical to the effect of vertex $v_j(t')$ causing interference
at vertex $v_i(t)$. What matters for the dynamic programming formulation in this section, is not whether $v_i(t)$ is
causing interference with $v_j(t')$ or vice versa, but whether $t > t'$, $t=t'$ or $t<t'$. In our algorithm we move
through vertices in order of increasing time $t$. {\it The state of the system at time $t$ is defined by those vertices
at times $\le t$ that are connected to vertices at times $\ge t$.}

More precisely, define the state vector at time $t$ to be
\begin{equation}\label{eqn:state_space}
{\bf s}(t) = [ v_1(t) \; \cdots \; v_1(t-l_1^*) \; \cdots \; v_K(t) \; \cdots \; v_K(t-l_K^*) ]^T.
\end{equation}

\noindent This is the collection of all vertices at times $\le t$, that interfere with, or are interfered with by
vertices at times $\ge t$. Figure \ref{fig:state_space} illustrates which vertices are included in the state vector. As
each $v_j(t)$ takes on one of two values, the state space consists of {\it at most} $2^{\sum_{j=1}^K (l_j^*+1)}$ possible
states. Some states may be infeasible because two of their vertices are connected by an edge. Thus define the state
space as the space of all feasible states
\begin{equation*}
{\cal S} = \left\{{\bf s}_1,\dots,{\bf s}_{|{\cal S}|}\right\}
\end{equation*}

\noindent with each ${\bf s}_i \in \{0,1\}^{\sum_{j=1}^K (l_j^*+1)}$. Each state in $\cal S$ corresponds to an {\it
independent set} in the subgraph made up of vertices in ${\bf s}(t)$. Thus there are typically far fewer than
$2^{\sum_{j=1}^K (l_j^*+1)}$ states. A second example is given in figures \ref{fig:dp_example1} and
\ref{fig:dp_example2}. In this example there are a total of 28 states as shown in figure \ref{fig:dp_example2}.

For notational convenience we label the elements of ${\bf s}_i$ as such
\begin{equation*}
{\bf s}_i = [s_i^{(1,0)} \; \cdots \; s_i^{(1,l_1^*)} \; \cdots \; s_i^{(K,0)} \; s_i^{(K,t-l_K^*)}]^T.
\end{equation*}

\noindent Denote the set of feasible state transitions from ${\bf a}\in{\cal S}$ to ${\bf b}\in{\cal S}$ by
\begin{equation*}
{\cal F} = \left\{({\bf a},{\bf b}) : b^{(j,k+1)}=a^{(j,k)}, \text{  for } j=1,\dots,K \text{ and } k = 0,\dots,l_j^*-1
\right\}
\end{equation*}

\begin{figure}
\centering
\includegraphics[width=350pt]{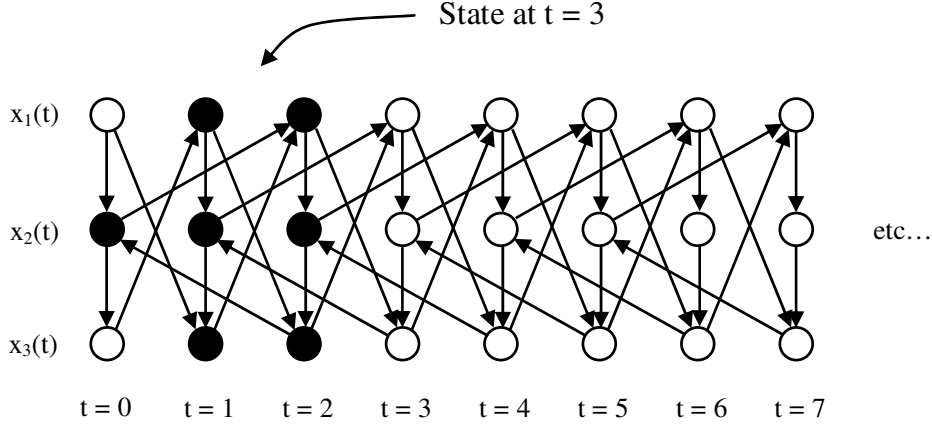}
\caption{A LOS interference channel with normalized cross delays $l'_{21} = 0$, $l'_{31}=1$, $l'_{12}=2$, $l'_{32}=0$,
$l'_{13}=1$ and $l'_{23}=-2$. The state at time $t=2$ is a function of the shaded vertices.} \label{fig:dp_example1}
\end{figure}

\begin{figure}
\centering
\includegraphics[width=420pt]{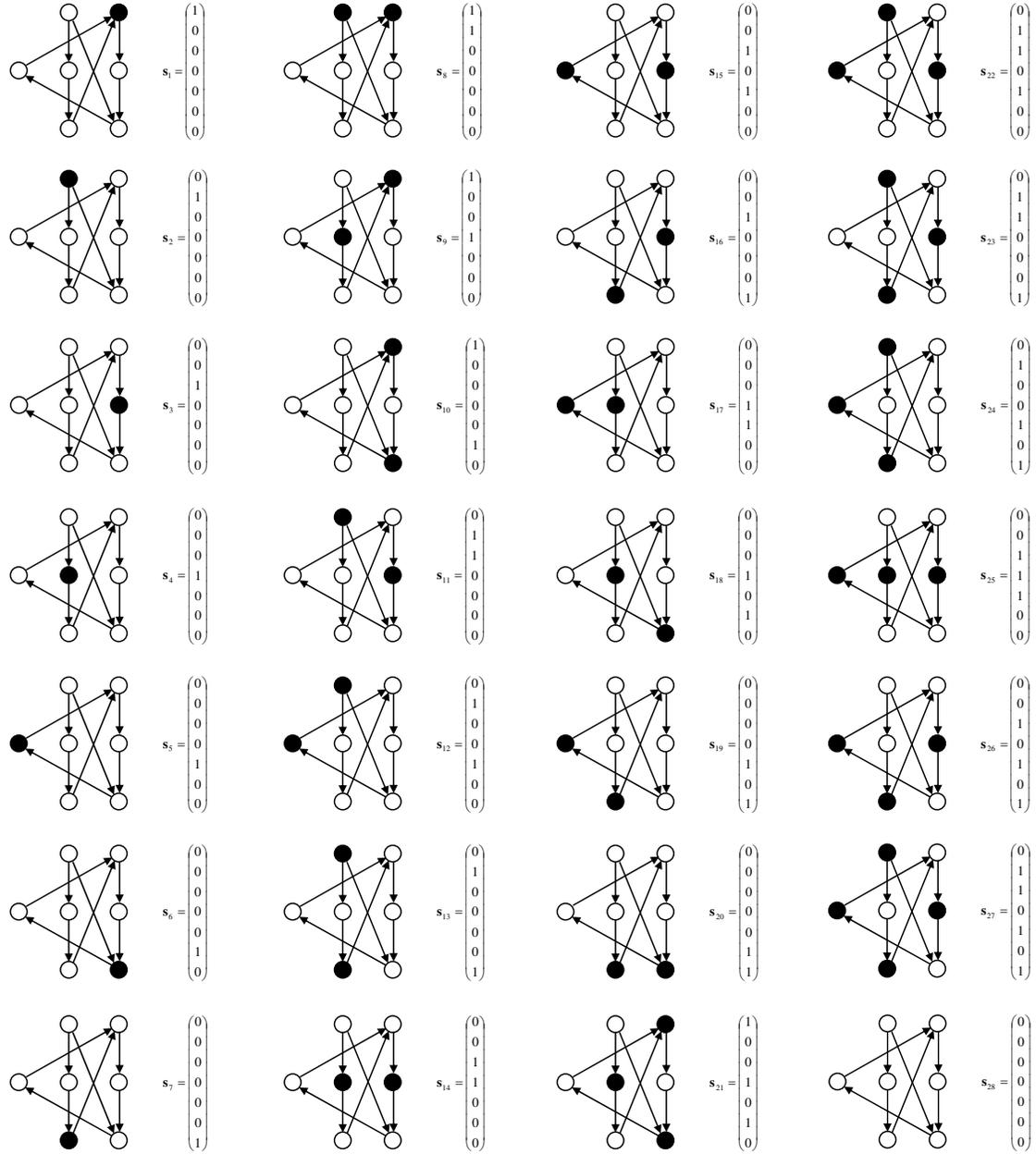}
\caption{An illustration of the entire state space $\cal S$ for the example of figure \ref{fig:dp_example1}. For each
state, both the state vector is given and the corresponding independent set shaded.} \label{fig:dp_example2}
\end{figure}

\noindent For the example of figures \ref{fig:dp_example1} and \ref{fig:dp_example2} the set of feasible state
transitions is
\begin{equation*}
{\cal S}= \{({\bf s}_1,{\bf s}_2),({\bf s}_1,{\bf s}_8),({\bf s}_1,{\bf s}_{11}),({\bf s}_2,{\bf s}_3),\dots\}
\end{equation*}

With the state space clearly defined, it is straightforward to derive the actual algorithm, which is akin to the
Viterbi algorithm. It begins by initializing the costs to zero, that is
\begin{equation*}
c_{{\bf s}_j}(0) = 0,
\end{equation*}

\noindent for all ${\bf s}_j \in {\cal S}$. Each iteration of the algorithm involves finding the minimum cost path
entering each state. At times $t=1,\dots,T$ compute
\begin{equation*}
{\bf s}_{i^*}(t,{\bf s}_j)=\text{arg}\min_{\scriptsize \begin{array}{c}
                                      {\bf s}_i \in {\cal S} \\
                                      ({\bf s}_i,{\bf s}_j) \in {\cal F}
                                    \end{array}} c_{{\bf s}_i}(t-1),
\end{equation*}

\noindent for each ${\bf s}_j\in {\cal S}$, which is the minimum cost state at time $t-1$ from which we can transition
into state ${\bf s}_j$ at time $t$. Then compute
\begin{equation*}
c_{{\bf s}_j}(t)= c_{{\bf s}_{i^*}(t,{\bf s}_j)}(t-1) -\sum_{k=1}^Kr_k s_j^{(k,0)},
\end{equation*}

\noindent for each ${\bf s}_j\in{\cal S}$, which is the minimum cost of a path that ends up at state ${\bf s}_j$ at
time $t$. When time $t=T$ is finally reached, compute
\begin{equation*}
{\bf s}^*(T)=\text{arg}\min_{{\bf s}_i \in {\cal S}} c_{{\bf s}_i}(T),
\end{equation*}

\noindent which is the optimal termination state. The optimal independent set is then found by working backwards. Start
by setting the $T$th column of vertices $v_1(T),\dots,v_K(T)$ according to ${\bf s}^*(T)$. That is, set $v_j(T) =
[s_i^*(T)]^{(j,0)}$ for $j=1,\dots,K$. Then set the $T-1$th column according to ${\bf s}_{i^*}(T,{\bf s}^*(T))$, that
is, set $v_j(T-1)=[s_i^*(T,{\bf s}^*(T))]^{(j,0)}$ for $j=1,\dots,K$. Continue by setting $v_j(T-2) =
[s_i^*(T-1,s_i^*(T,{\bf s}^*(T)))]^{(j,0)}$, etc...

We now briefly examine the complexity of this algorithm. As the state space consists of all independent sets of a
subgraph defined by $2^{\sum_{j=1}^K (l_j^*+1)}$ vertices, given a set of delays $l'_{ij}$, it takes $O(2^{\sum_{j=1}^K
(l_j^*+1)})$ time steps to enumerate. Once this is done, the algorithm takes $O(T|{\cal S}|)$ time steps to solve problem
(\ref{eqn:op_problem}). Roughly speaking, $l'_{ij}=O(L)$, and $|{\cal S}|=O(LK\log LK)$. Thus the algorithm takes
\begin{equation*}
O(TLK\log LK)+O(2^{LK})
\end{equation*}

\noindent time steps to compute the optimal independent set.

As the state space defined in equation \ref{eqn:state_space} is finite, for large $T$ the optimal independent set will
have a periodic form with period less than or equal to the number of states $|{\cal S}|$. Thus, if there is no
restriction on how large $T$ can be, once the period of the maximal independent set is found, we can simply set $T$
equal to it, without compromising optimality. In this case, the entire problem can be solved in
\begin{equation*}
O(LK\log LK2^{LK\log LK})+O\left(2^{LK}) = O((LK)^{LK+1}\log LK\right)
\end{equation*}

\noindent time steps.

\subsection{Bandwidth scaling}\label{sec:BS}

Theorem \ref{thm:mr} shows that if the bandwidth scales sufficiently quickly with $K$, the spectral efficiency
per user can be made to vanish arbitrarily slowly. A natural question to ask is whether it is {\it necessary} for the
bandwidth to scale with $K$, in order for this desirable property to hold. The following converse result establishes
that this is indeed the case.

\begin{theorem}\label{thm:W_scaling}
If the bandwidth scales sufficiently slowly with $K$ such that
\begin{equation*}
\lim_{K\rightarrow \infty}\frac{\log W}{\log\frac{K}{\log K}} = 0
\end{equation*}

\noindent then
\begin{equation}\label{eqn:alpha_scaling}
\lim_{K\rightarrow \infty} \frac{\log \alpha({\cal G})/T}{\log K} = 0
\end{equation}

\noindent with probability one.
\end{theorem}

\noindent As the total spectral efficiency
\begin{equation*}
\sum_{i=1}^K R_i\le \frac{\alpha({\cal G})}{T}\log_2\left(1+\max_i|h_{ii}|^2\frac{\PSD}{N_0}\right)
\end{equation*}

\noindent theorem \ref{thm:W_scaling} implies
\begin{equation*}
\lim_{K\rightarrow\infty}\frac{\log\left(\sum_{i=1}^K R_i\right)}{\log K} = 0,
\end{equation*}

\noindent which is equivalent to $\lim_{K\rightarrow \infty} R_i = 0$ for almost all users $i\in \{1,\dots,K\}$.

Roughly speaking the above result says that if the bandwidth scales slower than $O(K/\log K)$, then the spectral
efficiency resulting from any feasible transmit pattern will vanish as $K\rightarrow \infty$. Note there is a gap
between this converse result and the achievability result of theorem \ref{thm:mr}. Theorem \ref{thm:mr} demonstrates
that it is sufficient for the bandwidth to scale like $O((2K(K-1))^{K(K-1)})$, but theorem \ref{thm:W_scaling} shows
that it is necessary for the bandwidth to scale only as fast as $O(K/\log K)$. This establishes that the slowest
possible bandwidth scaling lies somewhere between $O(K/\log K)$ and $O((2K(K-1))^{K(K-1)})$. It is unclear, which if
any of these bounds is tight.

\begin{proof}
Remove those edges that connect vertices of different time slots in the interference graph, $(v_i(t),v_j(t'))$ for $t'
\neq t$. This provides an upper bound on the independence number. Now consider a single column ${\cal V}(t) =
\{v_1(t),\dots,v_K(t)\}$ of this graph in isolation. For any pair of vertices in ${\cal V}(t)$, there exists an edge
connecting them independently with probability $1-(1-1/L)^2$ (probability $1/L$ for each of the two possible
directions). Thus the graph consisting of vertices ${\cal V}(t)$ and the random subset of edges connecting them, is
precisely the Erd\H{o}s-R\'{e}nyi graph ${\cal G}_{K,1-(1-1/L)^2}$. A well known result (see for example
\cite{Bollobas}) is that
\begin{equation*}
\lim_{n\rightarrow \infty}\frac{\alpha({\cal G}_{n,p})}{\log n} = \frac{2}{\log(1/(1-p))}
\end{equation*}

\noindent with probability one. Hence the independence number of ${\cal G}_{K,2/L}$ satisfies
\begin{align*}
\lim_{n\rightarrow \infty}\frac{\alpha({\cal G}_{K,2/L})}{\log K} &= \frac{2}{\log(1/(1-1/L)^2)} \\
&\le L.
\end{align*}

\noindent As $L$ is directly proportional to $W$, if $\lim_{K\rightarrow \infty}\log W/\log\frac{K}{\log K} = 0$ then the same limit applies for $L$ and
\begin{equation*}
\lim_{K\rightarrow \infty} \frac{\log \alpha({\cal G}_{K,2/L})}{\log K} = 0.
\end{equation*}

\noindent Now as the independence number of the interference graph satisfies $\alpha({\cal G})\le T\alpha({\cal
G}_{K,2/L})$, equation (\ref{eqn:alpha_scaling}) follows.
\end{proof}

\subsection{When is the maximal independent set maximal?}

We now turn to the problem of analysis. Ideally we would like a simple characterization of the size of the maximal
independent set in terms of the parameters of the system, $K$, $T$, and the normalized cross-delays $l'_{ij}$. It is
unclear if such a characterization exists. Instead we present two results. The first characterizes when the
independence number is equal to its maximum possible value, and shows how the maximal independent set can be found
almost instantly in this event. The second, which is theorem \ref{thm:mr}, demonstrates that surprisingly large
independent sets exist on average, when both the number of users and the bandwidth are sufficiently high. In this
section we present the former result, in the next section we present the latter. What we will be revealed in this
section is that the problem of determining whether or not the independence number is equal to its maximum possible
value, is a group theoretic one.

We assume in this section that the direct gains are all equal so that the optimal independent set is equivalent to the
maximal independent set. Owing to the absence of some edges, the boundary of the interference graph has a slightly
different structure than it's interior. In order to circumvent this problem, we let $T\rightarrow \infty$ so that these
boundary effects are negligible.

\begin{definition}
The {\it independence rate} of sequence of interference graphs
\begin{equation*}
{\cal G}_{K,1}(\{l'_{ij}\}_{i\neq j}),{\cal G}_{K,2}(\{l'_{ij}\}_{i\neq j}),\dots
\end{equation*}
\noindent is
\begin{equation*}
{\sf IR}({\cal G}_K(\{l'_{ij}\}_{i\neq j})) \triangleq \lim_{T \rightarrow \infty}\frac{\alpha({\cal G}_{K,T}(\{l'_{ij}\}_{i\neq j}))}{T}.
\end{equation*}
\end{definition}

\noindent We write ${\sf IR}({\cal G}_K)$ for short. Start with the following observation.
\begin{lemma} For any number of users $K$ and any channel $\{l'_{ij}\}_{i\neq j}$
\begin{equation*}
{\sf IR}({\cal G}_K)\le \frac{K}{2}.
\end{equation*}
\end{lemma}

\noindent This means that for large $T$ we can only include at most half the vertices of the interference graph in any
feasible transmit pattern. We now ask, when is the independence rate {\it exactly equal} to $K/2$? The following result
succinctly answers this question for $K=3$. Define
\begin{align*}
l&\triangleq l'_{31}+l'_{13}+l'_{32}+l'_{21}+l'_{12}+l'_{23} \\
l_1 &\triangleq l'_{13}+l'_{32}+l'_{21}\\
l_2 &\triangleq l'_{21}+l'_{12}\\
l_3 &\triangleq l'_{31}+l'_{13}.
\end{align*}

\noindent If $l_i\neq 0$, define $\gamma_i$ to be the exponent of 2 in the prime factorization of $l_i$, that is
$l_i=2^{\gamma_i}\beta_i$ where $\beta_i$ represents the rest of the prime factorization. If $l_i=0$ then define $\gamma_i=\infty$. Similarly if $l\neq 0$, define $\gamma$ to
be the exponent of 2 in the prime factorization $l$, i.e. $l=2^{\gamma}\beta$. If $l=0$ then define $\gamma=\infty$.

\begin{theorem}\label{thm:C}
${\sf IR}({\cal G}_3)=3/2$ if and only if $\gamma_1 < \gamma_2$, $\gamma_1 < \gamma_3$ and $\gamma_1 < \gamma$, in
which case there are exactly $2^{\gcd(l_1,l_2/2,l_3/2,l/2)}$ feasible transmit patterns achieving it.
\end{theorem}

To clarify, if for example both $\gamma_1=\infty$ and $\gamma_2=\infty$, then the above conditions are not satisfied and ${\sf IR}({\cal G})<3/2$. This theorem provides a necessary and sufficient condition such that all users can transmit half the time without interfering with one another. The most probable way this condition can be met is if $l_1$ is an odd number, and $l,l_2$
and $l_3$ are all even numbers. Each of these events roughly occurs independently with probability $1/2$, hence the
probability all four occur simultaneously is $1/16$. Thus with probability $\gtrsim 1/16$ there exists a feasible
transmit pattern enabling all users to transmit half the time without interfering with one another. There are of course
other ways in which our condition can be met, for example, if $l_1$ is even, but not a multiple of 4, and $l,l_2$ and
$l_3$ are all multiples of 4, however this, and all other configurations satisfying the condition of theorem
\ref{thm:C} likely occur with probability much less than $1/16$. The proof of theorem \ref{thm:C} is given in the
appendix.

\begin{example}
For the channel in figure \ref{fig:chain_graphs}, we have $l=2$ and $l_1 = 1, l_2 = 2, l_3 = 2$. This means $\gamma =
1,\gamma_1=0,\gamma_2=1,\gamma_3=1$, so a feasible transmit pattern achieving independence number $3/2$ exists. As
$\gcd(l_1,l_2/2,l_3/2,l/2)=1$ there are only two feasible transmit patterns: the first is shown in the figure as a
sequence of shaded vertices, the second is obtained by complementing the transmit pattern, i.e. unshading the shaded
vertices, and shading the unshaded ones.
\end{example}

How does theorem \ref{thm:C} generalize for an arbitrary number of users, $K$? Define a {\it cycle} on the interference
graph to be the indices of a tuple of edges with connecting vertices, that start and end on the same row. For example
$((1,2),(2,3),(3,1))$ and $((3,2),(2,3))$ are examples of cycles for $K=3$. The {\it length} of a cycle is the sum of
the normalized cross-delays associated with it. For example, the cycle $((1,2),(2,3),(3,1))$ has length
$l'_{12}+l'_{23}+l'_{31}$.

Define the set of cycles containing an even number of terms as
\begin{equation*}
{\cal Y}_e =  \left\{ \left( (i_1,i_2),(i_2,i_3),\dots,(i_{2n},i_1) \right) : i_1 \neq i_2 \neq \dots \neq i_{2n}
\text{ and } i_j\in\{1,\dots,K\}\right\},
\end{equation*}

\noindent and the set of cycles containing an odd number of terms as
\begin{equation*}
{\cal Y}_o =  \left\{ \left( (i_1,i_2),(i_2,i_3),\dots,(i_{2n+1},i_1) \right) : i_1 \neq i_2 \neq \dots \neq i_{2n+1}
\text{ and } i_j\in\{1,\dots,K\}\right\}.
\end{equation*}

\noindent Then we claim it can be shown that
\begin{claim}\label{thm:chain_graph_result}
${\IR}({\cal G}_K)=K/2$ if and only if the exponents of 2 in the prime factorizations of the lengths of all cycles
containing an odd number of terms, are the same, and this exponent is strictly less than the exponent of 2 in the prime
factorization of the length of every cycle containing an even number of terms. That is, the exponent of 2 in the prime
factorization of
\begin{equation*}
\sum_{(i,j)\in Y_o} l'_{ij},
\end{equation*}

\noindent is the same for all $Y_o \in {\cal Y}_o$, and this value is strictly less than the exponent of 2 in the prime
factorization of
\begin{equation*}
\sum_{(i,j)\in Y_e} l'_{ij},
\end{equation*}

\noindent for any $Y_e\in {\cal Y}_e$.
\end{claim}

\section{Achieving Non-Vanishing Spectral Efficiency}\label{sec:non_vanish_se}

We now prove theorem \ref{thm:mr} by presenting a construction with expected spectral efficiency that can be made to
vanish arbitrarily slowly as $K\rightarrow \infty$. First, a high-level overview of the proof. The idea is to construct
a transmit pattern that has close to $O(1)$ independence rate as $K\rightarrow \infty$. At the heart of the transmit
pattern is a {\it generalized arithmetic progression}. If the bandwidth scales appropriately with the number of users
then this progression will have desirable interference alignment properties, but care has to be taken in constructing a
transmit pattern out of it. In particular, the progression will be very sparse, meaning that many identical versions of
this progression must be interleaved, each with a different timing offset. The trick to making the analysis work is to
use a randomization argument to show that a good set of offsets exists.

\begin{proof}
(of Theorem \ref{thm:mr}) First some notation that will be used throughout the proof. Let
\begin{equation*}
N\triangleq K(K-1)
\end{equation*}

\noindent and
\begin{equation*}
A \triangleq \frac{L}{N^{N+\epsilon}}.
\end{equation*}

\noindent We use $\oplus$ to denote addition modulo $L$. Let
\begin{equation*}
{\cal T} \triangleq \left\{ \bigoplus_{1\le i \neq j \le K} \alpha_{ij}l_{ij} : \{\alpha_{ij}\}_{i\neq j} \in
\{0,\dots,N-1\}^N \right\}.
\end{equation*}

\noindent This is the set of all linear combinations of the cross-delays (not normalized) with integer coefficients
ranging from 0 to $N$. Define
\begin{equation*}
{\cal S} \triangleq \bigcup_{a=1}^{A}(m_a \oplus {\cal T})
\end{equation*}

\begin{figure}
\centering
\includegraphics[width=400pt]{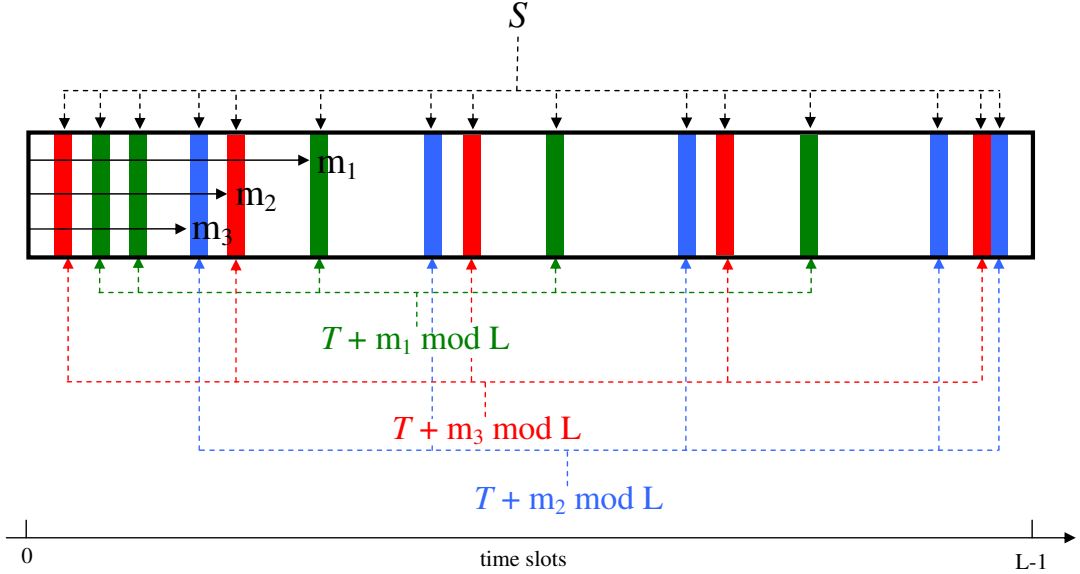}
\caption{The construction ${\cal S}$ is formed by interleaving a sufficient number of generalized arithmetic
progressions $\cal T$ with random offsets $m_i$. The colored bars indicate time slots during which data symbols are
sent.} \label{fig:interleaving}
\end{figure}

\noindent See figure \ref{fig:interleaving} for an illustration. Each user transmits one data symbol at each time slot in the set
\begin{equation*}
{\cal X} \triangleq \bigcup_{k=0}^\infty ({\cal S}+kL).
\end{equation*}

\noindent The above construction corresponds to concatenating data blocks $\cal S$ of length $L$. The modulo $L$
addition used in constructing $\cal S$ and $\cal T$, ensures a seamless transition at the block boundaries. This is
illustrated in figure \ref{fig:blocks}. The construction is defined in this seemingly convoluted way in order to make
the analysis simple and elegant. However there is an easier way of conceptualizing this construction: take multiple
copies of the generalized arithmetic progression $\left\{ \sum_{1\le i \neq j \le K} \alpha_{ij}l_{ij} :
\{\alpha_{ij}\}_{i\neq j} \in \{0,\dots,N-1\}^N \right\}$, and throw them down on the infinite time axis with offsets
$m_1,\dots,m_A,L+m_1,\dots,L+m_A,2L+m_1,\dots,2L+m_A,\dots$. Although this construction is periodic with period $L$,
locally, the offsets of these progressions will appear as a Poisson process with intensity $A/L = 1/N^{N+\epsilon}$. As
there are $N^N$ points in each progression, the density of points in $\cal X$ will be $1/N^{\epsilon}$ and hence the
spectral efficiency will go to zero with $K$ like $1/(K(K-1))^\epsilon$.

We will show that there exists a choice of
\begin{equation*}
(m_1,\dots,m_A) \in \{0,\dots,L-1\}^A
\end{equation*}

\noindent such that the expected spectral efficiency of this scheme approaches (\ref{eqn:Erate}) as
$K\rightarrow\infty$. More specifically, we show that for the above construction, at each receiver the expected
fraction of time slots containing a data symbol but no interference is large. Each such data symbol is then able to
convey $\log_2({1 + |h_{ii}|^2\PSD/N_0})$ bps/Hz of information and the expected spectral efficiency achieved by the
scheme for user $i$ is the fraction of such time slots multiplied by $\log_2({1 + \PSD/N_0})$.

Since our construction $\cal X$ consists of a concatenation of identical blocks of length $L$, we analyze its
performance over a single block extending from time slot $0$ to $L-1$. At receiver $i$ the set of time slots containing
interference is
\begin{equation}\label{eqn:Fi}
{\cal F}_i \triangleq \bigcup_{j=1, j\neq i}^K \left( {\cal S} \oplus l_{ij}\right).
\end{equation}

\begin{figure}
\centering
\includegraphics[width=420pt]{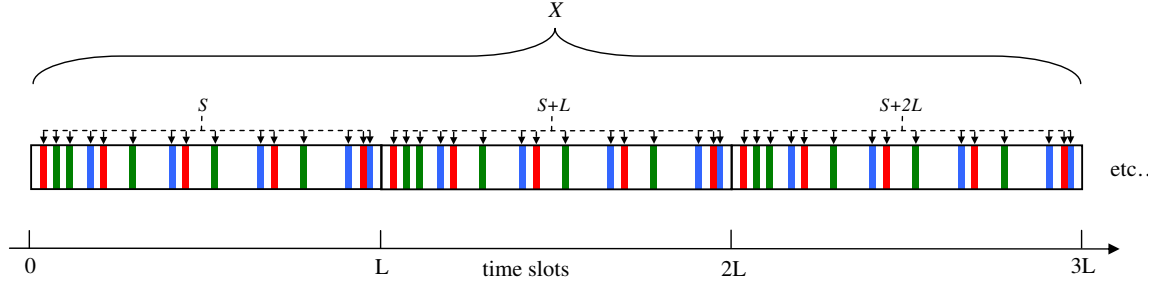}
\caption{The construction ${\cal X}$ is formed by concatenating blocks of $\cal S$. The modulo $L$ structure of $\cal S$ ensures a seamless transition from block to block.} \label{fig:blocks}
\end{figure}

Denote the number of time slots at receiver $i$, that contain a data symbol from transmitter $i$, but no interference
by $S_i \triangleq \left| \{ t\in ({\cal S}\oplus l_{ii})\backslash{\cal F}_i \} \right|$. Then conditioning on the
cross delays we have
\begin{equation}\label{eqn:ENi}
\mathbb E S_i = \frac{1}{L^{K(K-1)}}\sum_{\{l_{ij}\}_{i\neq j}}\mathbb E_{\{l_{ii}\}_i} [S_i | \{l_{ij}\}_{i\neq j}].
\end{equation}

\noindent Define
\begin{equation*}
s(k) \triangleq \left\{
         \begin{array}{ll}
           1, & \hbox{if $k\in {\cal S}$} \\
           0, & \hbox{otherwise}.
         \end{array}
       \right.
\end{equation*}

\noindent That is, $s(k)=1$ if a transmission takes place at time slot $k$, and zero otherwise. Similarly define
\begin{equation*}
f_i(k) \triangleq \left\{
         \begin{array}{ll}
           1, & \hbox{if $k \in {\cal F}_i$} \\
           0, & \hbox{otherwise}.
         \end{array}
       \right.
\end{equation*}

\noindent That is, $f_i(k)=1$ if there is interference during time slot $k$ at receiver $i$ and zero otherwise. Conditioned
on the cross delays, $S_i$ is the correlation function between the set of transmit times $\cal S$ and the set of
interference free times ${\cal F}_i^c$, evaluated at an offset of $l_{ii}$, specifically
\begin{align*}
S_i(l_{ii}) &= \sum_{k=0}^{L-1} s(k\oplus l_{ii})(1-f_i(k))
\end{align*}

\noindent Thus
\begin{align*}
E_{\{l_{ii}\}_i} [S_i | \{l_{ij}\}_{i\neq j}] &= \frac{1}{L}\sum_{l_{ii}=0}^{L-1} S_i(l_{ii}) \\
&= \frac{1}{L}\sum_{l_{ii}=0}^{L-1}\sum_{k=0}^{L-1} s(k\oplus l_{ii})(1-f_i(k)) \\
&= \frac{1}{L}\sum_{k=0}^{L-1} (1-f_i(k))\sum_{l_{ii}=0}^{L-1}s(k\oplus l_{ii}) \\
&= \frac{1}{L}\sum_{k=0}^{L-1} (1-f_i(k))\sum_{l_{ii}=0}^{L-1}s(l_{ii}) \\
&= \frac{1}{L}\sum_{k=0}^{L-1} (1-f_i(k))|{\cal S}| \\
&= |{\cal S}|\left(1-\frac{|{\cal F}_i|}{L}\right)
\end{align*}

\noindent where in the above sequence of equations we have used the identities $|{\cal S}| \equiv \sum_{k=0}^{L-1}s(k)$
and $|{\cal F}_i| \equiv \sum_{k=0}^{L-1}f_i(k)$. Substituting back into equation (\ref{eqn:ENi}) we find the fraction
of time slots at receiver $i$ containing data but no interference is
\begin{align}\label{eqn:ESoverL}
\frac{\mathbb ES_i}{L} &= \frac{1}{L^{K(K-1)}}\sum_{\{l_{ij}\}_{i\neq j}} \frac{|{\cal S}|}{L}\left(1-\frac{|{\cal
F}_i|}{L}\right).
\end{align}

\noindent The above expression makes intuitive sense as if we uniformly select a time slot at random from
$\{0,\dots,L-1\}$, then conditioned on the $\{l_{ij}\}_{i\neq j}$, the quantity $|{\cal S}|/L$ is the probability this
time slot contains a data symbol, and $1-|{\cal F}_i|/L$ is the probability it does not contain an interference
symbol. We now compute appropriate bounds on the terms $|{\cal S}|$ and $|{\cal F}_i|$. From equation (\ref{eqn:Fi}) we have
\begin{align*}
{\cal F}_i &= \bigcup_{j\neq i} ({\cal S} \oplus l_{ij}) \\
&= \bigcup_{j\neq i} \left( \left[\bigcup_{a=1}^A(m_a\oplus{\cal T})\right] \oplus l_{ij} \right) \\
&= \bigcup_{a=1}^A  \left[ \bigcup_{j\neq i}({\cal T} \oplus l_{ij}) \right] \oplus m_a
\end{align*}

\noindent But
\begin{equation*}
\bigcup_{j\neq i}({\cal T} \oplus l_{ij}) \subset \left\{ \bigoplus_{1\le i \neq j \le K} \alpha_{ij}l_{ij} : \{\alpha_{ij}\}_{i\neq j} \in
\{0,\dots,N\}^N \right\},
\end{equation*}

\noindent and this set has at most $(N+1)^N$ elements. This is the interference alignment property. Hence
\begin{align*}
\frac{|{\cal F}_i|}{L} &\le A \frac{1}{L}\left| \bigcup_{j\neq i} ( {\cal T} \oplus l_{ij} ) \right| \\
&\le \frac{A(N+1)^N}{L} \\
&< \frac{(N+1)^N}{N^{N+\epsilon}} \\
&= N^{-\epsilon}\left(1+\frac{1}{N}\right)^N \\
&< e N^{-\epsilon}.
\end{align*}

We now bound $|{\cal S}|$. We first show that $|{\cal T}| = N^N$ almost surely as $N\rightarrow \infty$. In order to have $|{\cal T}| < N^N$, there must exist two sets of coefficients $\{\alpha_{ij}\}_{i\neq j}\neq
\{\alpha_{ij}\}'_{i\neq j}$ both elements of $\{0,\dots,N-1\}^N$, satisfying
\begin{equation*}
\bigoplus_{i\neq j} \alpha_{ij}l_{ij} = \bigoplus_{i\neq j} \alpha_{ij}'l_{ij}.
\end{equation*}

\noindent This is equivalent to requiring there to exist some $\{\overline \alpha_{ij}\}_{i \neq j} \in
\{-N+1,\dots,N-1\}^N\backslash {\bf 0}$ satisfying
\begin{equation*}
\bigoplus_{i\neq j} \overline \alpha_{ij} l_{ij} = 0.
\end{equation*}

\noindent Using the union bound we have
\begin{align*}
\Pr\left(|{\cal T}|<N^N\right) &= \Pr \left( \exists \{\overline \alpha_{ij}\}_{i\neq j} \in
\{-N+1,\dots,N-1\}^N\backslash {\bf 0} \text{ s.t. } \bigoplus_{i\neq j}\overline \alpha_{ij}l_{ij} =0 \right) \\
&\le \sum_{\{\overline \alpha_{ij}\}_{i\neq j} \in \{-N+1,\dots,N-1\}^N} \Pr \left( \bigoplus_{i\neq j}\overline
\alpha_{ij}l_{ij} =0 \right).
\end{align*}

\noindent As conditioned on all cross delays other than $l_{12}$, there is at most one value of $l_{12}$ that satisfies
$\sum_{i\neq j}\overline \alpha_{ij}l_{ij} =0$, and the cross delays are uniformly distributed over $\{0,\dots,L-1\}$,
we have
\begin{align*}
\Pr\left(|{\cal T}|<N^N\right) &\le \sum_{\{\overline \alpha_{ij}\}_{i\neq j} \in \{-N+1,\dots,N-1\}^N}
\frac{1}{L} \\
&= \frac{(2N)^N}{L} \\
&\le \frac{(2N)^N}{(2N)^{N+\epsilon}} \\
&= (N)^{-\epsilon} \\
&\rightarrow 0
\end{align*}

\noindent as $N\rightarrow \infty$ (or equivalently, as $K\rightarrow \infty$).

At this point it should start to become clear why it is that in this particular construction the bandwidth must scale
like $(2N)^N$. From the above calculation we see that in order to make all the points in the generalized arithmetic
progression $\cal T$ distinct, we require the $l_{ij}$ to be large. How large? It may seem that as there are $N^N$
integers in $\cal T$, the minimum being 0 and the maximum being roughly the same order as $l_{ij}$, we require $l_{ij}
= O(N^N)$. However, the structure of the generalized arithmetic progression is such that the bulk of its points are
concentrated around the center, such that we actually require at least $l_{ij}=O((2N)^N)$ to separate these center
points out, as the above calculation shows. But such a large order of $l_{ij}$ makes $\cal T$ very sparse, in fact if
$l_{ij} = O((2N)^N)$ then $\cal T$'s density is a mere $O(2^{-N})$. So to fill in the gaps we interleave multiple
sequences $\cal T$. How many? $O(2^N)$. In general if $l_{ij} = O(L)$ then we must interleave $O(L/N^N)=A$ sequences.
This explanation is illustrated in figure \ref{fig:distinct}

\begin{figure}
\centering
\includegraphics[width=420pt]{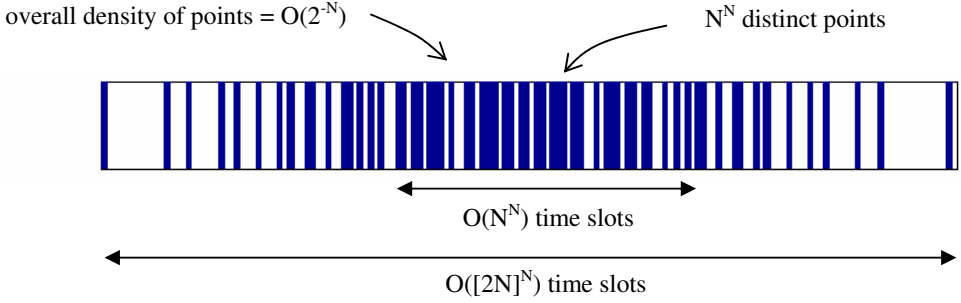}
\caption{An illustration of the generalized arithmetic progression $\cal T$ without the modulo $L$ wrap around. There
are $N^N$ points spread over $O((2N)^N)$ time slots, however almost all points are concentrated at the center in a
width of $O(N^N)$. Hence the density of points is $O(2^{-N})$.} \label{fig:distinct}
\end{figure}

We now use a probabilistic argument to demonstrate the existence of a good choice of $(m_1,\dots,m_A)$. Let $m_a\sim$
i.i.d. $U(\{0,\dots,L-1\})$. Let
\begin{equation*}
t(k) \triangleq \left\{
                  \begin{array}{ll}
                    1, & \hbox{if $k\in {\cal T}$} \\
                    0, & \hbox{otherwise.}
                  \end{array}
                \right.
\end{equation*}

\noindent and
\begin{equation*}
t_a(k) \triangleq \left\{
                  \begin{array}{ll}
                    1, & \hbox{if $k\in m_a\oplus{\cal T}$} \\
                    0, & \hbox{otherwise.}
                  \end{array}
                \right.
\end{equation*}

\noindent Thus $t_a(k) = t(k\oplus m_a)$. Then we can write
\begin{multline*}
s(k) = t_1(k) + t_2(k)(1-t_1(k)) + t_3(k)(1-t_2(k))(1-t_1(k)) \\
+ \dots + t_A(k)(1-t_{A-1}(k))\times\dots \times(1-t_1(k)).
\end{multline*}

\noindent This expression says that $k\in {\cal S}$ if $k\in m_1\oplus{\cal T}$, or if $k\notin m_1\oplus{\cal T}$ but $k\in m_2\oplus{\cal T}$, or if $k\notin m_1\oplus{\cal T}$ and $k\notin m_2\oplus {\cal T}$ but $k\in m_3\oplus{\cal T}$, etc... We can write the above expression alternatively as
\begin{align*}
s(k) &= t_1(k)+(1-t_1(k))(t_2(k)+(1-t_2(k))(t_3(k)+\dots)) \\
&= t(k\oplus m_1)+(1-t(k\oplus m_1))(t(k\oplus m_2)+(1-t(k\oplus m_2))(t(k\oplus m_3)+\dots)) \\
\end{align*}

\noindent Then taking the expectation over the distribution of $m_1,\dots,m_A$
\begin{align*}
\frac{\mathbb E|{\cal S}|}{L} &= \frac{1}{L} \mathbb E \sum_{k=0}^{L-1} s(n) \\
&= \frac{1}{L} \mathbb E \sum_{k=0}^{L-1} t(k\oplus m_1)+(1-t(k\oplus m_1))(t(k\oplus m_2) \\
&\quad\quad\quad\quad\quad\quad\quad\quad\quad\quad\quad\quad\quad\quad +(1-t(k\oplus m_2))(t(k\oplus m_3)+\dots)) \\
&= \frac{1}{L^{A+1}} \sum_{m_A=0}^{L-1} \cdots \sum_{m_1=0}^{L-1} \sum_{k=0}^{L-1} t(k\oplus m_1)+(1-t(k\oplus m_1))(t(k\oplus m_2) \\
&\quad\quad\quad\quad\quad\quad\quad\quad\quad\quad\quad\quad\quad\quad +(1-t(k\oplus m_2))(t(k\oplus m_3)+\dots)) \\
&= \frac{1}{L^{A+1}} \sum_{k=0}^{L-1}\sum_{m_A=0}^{L-1} \cdots \sum_{m_1=0}^{L-1} t(k\oplus m_1)+(1-t(k\oplus m_1))(t(k\oplus m_2) \\
&\quad\quad\quad\quad\quad\quad\quad\quad\quad\quad\quad\quad\quad\quad +(1-t(k\oplus m_2))(t(k\oplus m_3)+\dots)) \\
&= \frac{1}{L^{A+1}} \sum_{k=0}^{L-1}\sum_{m_A=0}^{L-1} \cdots \sum_{m_2=0}^{L-1}|{\cal T}|+(L-|{\cal T}|)t(k\oplus m_2) \\
&\quad\quad\quad\quad\quad\quad\quad\quad\quad\quad\quad\quad\quad\quad +(L-|{\cal T}|)(1-t(k\oplus m_2))(t(k\oplus m_3)+\dots)) \\
&= \frac{1}{L^{A+1}} \sum_{k=0}^{L-1}\sum_{m_A=0}^{L-1} \cdots \sum_{m_2=0}^{L-1}|{\cal T}|+|{\cal T}|(L-|{\cal T}|) \\
&\quad\quad\quad\quad\quad\quad\quad\quad\quad\quad\quad\quad\quad\quad +(L-|{\cal T}|)^2(t(k\oplus m_3)+\dots)) \\
&\quad\quad\quad\quad\quad\quad\quad\quad\quad\quad\quad\quad\quad\quad \vdots \\
&= \frac{|{\cal T}|}{L}\sum_{a=0}^{A-1}\left( 1 - \frac{|{\cal T}|}{L} \right)^a \\
&= 1 - \left(1-\frac{|{\cal T}|}{L}\right)^A \\
&\rightarrow 1 - \left(1-\frac{N^N}{L}\right)^{\frac{L}{N^{N+\epsilon}}} \quad\quad \text{ a.s.} \\
&\rightarrow 1 - e^{-\frac{N^N}{L}\frac{L}{N^{N+\epsilon}}} \\
&\rightarrow N^{-\epsilon}.
\end{align*}

\noindent Substituting back into equation (\ref{eqn:ESoverL})
\begin{align*}
\frac{\mathbb ES_i}{L} &= \frac{1}{L^{K(K-1)}}\sum_{\{l_{ij}\}_{i\neq j}} N^{-\epsilon}\left(1-eN^{-\epsilon}\right) \\
&= N^{-\epsilon}\left(1-eN^{-\epsilon}\right) \\
&\rightarrow N^{-\epsilon}
\end{align*}

\noindent as $N\rightarrow \infty$ (or equivalently $K\rightarrow \infty$). As each data symbol that is received
without interference is capable of reliably communicating $\log_2(1+|h_{ii}|^2{\sf PSD/N_0})$ bps/Hz, the expected
spectral efficiency of each user $i$ goes to
\begin{equation*}
\frac{1}{(K(K-1))^{\epsilon}} \log_2\left(1+|h_{ii}|^2{\sf PSD/N_0}\right)
\end{equation*}

\noindent as $K\rightarrow \infty$.

\end{proof}

The proof of theorem \ref{thm:mr2} is a straightforward extension of the previous.
\begin{proof}
(of Theorem \ref{thm:mr2}) Assume without loss of generality that $\arg \max_{d} |h_{ii,d}|^2 = 1$ for all users $i$. If each receiver $i$ treats physical paths $2,3,\dots,D$ from transmitter $i$ as
interference, then the received signals in the $K$ user $D$ path interference channel are statistically identical to
those of the $DK$-user LOS interference channel. Thus the achievability result of theorem \ref{thm:mr} carries
over to the $D$ path channel with $K$ replaced by $DK$, and $|h_{ii}|^2$ replaced by $\max_{d \in
\{1,\dots,D\}}|h_{ii,d}|^2$.
\end{proof}

\section{Frequency Domain Interpretation}\label{sec:fdomain}

In this section we reconcile the time domain version of interference alignment presented in this paper, and the
frequency domain results of \cite{Jafar1}. Specifically we show how the three-user construction of \cite{Jafar1} has a
simple time domain structure for the LOS interference channel.

To begin, we need to transform the LOS model into the frequency domain. For this we use an OFDM architecture summarized
in figure \ref{fig:ofdm}. Transmitter $j$ has a stream of complex data symbols to send $\{x_j[0],x_j[1],\dots\}$ to
receiver $i$. These are broken up into blocks of length $n$. Consider a single block denoted ${\bf x} =
[x_j[0],\dots,x_j[n-1]]^T$. To send this block the transmitter computes the $M$-length vector $\overline {\bf x}_j =
{\bf V}_j{\bf x}_j$, where ${\bf V}_j \in {\mathbb C}^{M\times n}$ is an encoding matrix to be specified later. Let
$l_{\max} \triangleq \max_{i,j} l_{ij}$. To send $\overline {\bf x}_j$, tx $j$ computes its IDFT and appends a cyclic
prefix of length $l_{\max}$. Each receiver removes the cyclic prefix and computes the DFT. Specifically
\begin{equation*}
\tilde {\bf x}_j = \left(
                     \begin{array}{c}
                     {\bf 0}_{l_{\max}\times (M-l_{\max})} \;\;{\bf I}_{l_{\max}\times l_{\max}} \\
                     {\bf I}_{M\times M}
                     \end{array}
                   \right) {\bf F}_{M\times M}^*\overline {\bf x}_j
\end{equation*}

\noindent
and
\begin{equation*}
{\bf y}_j = {\bf F}_{M\times M}\left(
                     \begin{array}{c}
                     {\bf 0}_{M\times l_{\max}} \;\;{\bf I}_{M\times M}
                     \end{array}
                   \right) \tilde {\bf y}_j
\end{equation*}

\noindent where ${\bf F}_{M\times M}$ is the $M\times M$ DFT matrix,
\begin{equation*}
{\bf F}_{M\times M} = \frac{1}{\sqrt{M}} \left(
            \begin{array}{ccccc}
              1 & 1 & 1 & \cdots & 1 \\
              1 & e^{j2\pi/M} & e^{j2\pi \cdot 2/M} & \cdots & e^{j2\pi\cdot (n-1)/M} \\
              1 & e^{j2\pi\cdot 2/M} & e^{j4\pi\cdot (n-1)\theta} & \cdots & e^{j2\pi\cdot 2(n-1)/M} \\
              \vdots & \vdots & \vdots & \ddots & \vdots \\
              1 & e^{j2\pi\cdot (M-1)/M} & e^{j2\pi\cdot 2(M-1)\theta} & \cdots & e^{j2\pi\cdot (M-1)(n-1)/M} \\
            \end{array}
          \right).
\end{equation*}

\begin{figure}
\centering
\includegraphics[width=420pt]{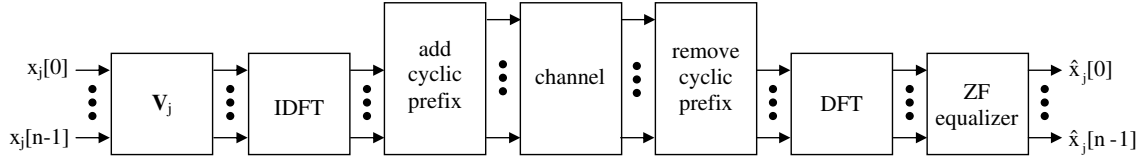}
\caption{Illustration of the OFDM architecture used to reconcile the time and frequency domain versions of interference
alignment.} \label{fig:ofdm}
\end{figure}

\noindent The result is a length $M$ sequence
\begin{equation*}
y_i[k] = \sum_{j=1}^K h_{ij} e^{-j2\pi (f_c \tau_{ij}+kl_{ij}/M)}\overline x_j[k] + z_i[k]
\end{equation*}

\noindent for $k=0,\dots,M-1$. Let $\theta_{ij} \triangleq l_{ij}/M$. Then in matrix form
\begin{equation*}
{\bf y}_i = \sum_{j=1}^K {\bf H}_{ij}{\bf V}_j{\bf x}_j + {\bf z}_i
\end{equation*}

\noindent where the link matrices are
\begin{equation}\label{eqn:ch_matrices}
{\bf H}_{ij}=h_{ij}e^{-j2\pi f_c\tau_{ij}}\left(
              \begin{array}{ccccc}
                1 &  &  &  &  \\
                 & e^{-j2\pi\cdot \theta_{ij}} &  &  &  \\
                 &  & e^{-j2\pi\cdot 2\theta_{ij}} &  &  \\
                 &  &  & \ddots &  \\
                 &  &  &  & e^{-j2\pi\cdot (M-1)\theta_{ij}} \\
              \end{array}
            \right).
\end{equation}

\noindent Choose $W$ sufficiently large such that the $l_{ij}$ are distinct. Note that $M$ needs to be much larger than $l_{\max}$ in order for the overhead from the cyclic prefix to be small. Also, $M$ must not have any of the $l_{ij}$ as divisors, else the channel matrices will lose rank. Let $\overline {\bf H} \triangleq h_{ij}^{-1}e^{j2\pi f_c \tau_{ij}}{\bf H}$ and ${\bf T}
\triangleq \overline {\bf H}_{12}\overline {\bf H}_{21}^{-1}\overline {\bf H}_{23}\overline {\bf H}_{32}^{-1}\overline
{\bf H}_{31}\overline {\bf H}_{13}^{-1}$. Let ${\bf w} = [1 \; \cdots \; 1]^T$. Choose the encoding matrices as follows
\begin{align*}
{\bf V}_1 &= \overline {\bf H}_{31}^{-1} \overline{\bf H}_{32}{\bf T}^2{\bf V} \\
{\bf V}_2 &= {\bf T} {\bf V} \\
{\bf V}_3 &= \overline {\bf H}_{13}^{-1}\overline {\bf H}_{12} {\bf V}
\end{align*}

\noindent where
\begin{equation*}
{\bf V} = [{\bf w} \; {\bf Tw}\; {\bf T}^2{\bf w} \; \cdots \; {\bf T}^{n-1}{\bf w}]
\end{equation*}

\noindent and let
\begin{equation*}
n=(M-1)/2.
\end{equation*}

\noindent where $M$ will be chosen to be an odd number. This is the three user construction of \cite{Jafar1}, but note the channel matrices ${\bf H}_{ij}$ do not consist of $M$ independently faded tones. Rather, all tones are derived from a single parameter $l_{ij}$. Define
\begin{align*}
\theta &= \theta_{12} - \theta_{21} + \theta_{23} - \theta_{32} + \theta_{31} - \theta_{13} \\
l &= l_{12} - l_{21} + l_{23} - l_{32} + l_{31} - l_{13}
\end{align*}

\noindent Then
\begin{equation}\label{eqn:V}
{\bf V} = \frac{1}{\sqrt{M}}\left(
            \begin{array}{ccccc}
              1 & 1 & 1 & \cdots & 1 \\
              1 & e^{j2\pi\cdot \theta} & e^{j2\pi \cdot 2\theta} & \cdots & e^{j2\pi\cdot (n-1)\theta} \\
              1 & e^{j2\pi\cdot 2\theta} & e^{j2\pi\cdot 4\theta} & \cdots & e^{j2\pi\cdot 2(n-1)\theta} \\
              \vdots & \vdots & \vdots & \ddots & \vdots \\
              1 & e^{j2\pi\cdot (M-1)\theta} & e^{j2\pi\cdot 2(M-1)\theta} & \cdots & e^{j2\pi\cdot (M-1)(n-1)\theta} \\
            \end{array}
          \right).
\end{equation}

\begin{lemma}\label{lem:V}
If $M$ is prime the columns of $\bf V$ are a permuted subset of the columns of ${\bf F}_{M\times M}$, i.e.
\begin{equation}
{\bf V} = {\bf F}_{M\times M}{\bf \pi}_{l,M}
\end{equation}

\noindent where ${\bf \pi}_{l,M}$ is an $M\times n$ permutation matrix, i.e. each column of ${\bf \pi}_{l,M}$ is a unique column of ${\bf I}_{M\times M}$.
\end{lemma}

\begin{proof}
Consider the matrix element ${\bf V}(2,k) = e^{j2\pi kl/M}/\sqrt{M} =  e^{j2\pi \cdot (kl\mod M)/M }/\sqrt{M}$ for some
$k\in\{0,\dots,M-1\}$. Let $M$ be prime. Then the set
\begin{equation*}
\{1,e^{j2\pi/M},e^{j2\pi\cdot 2/M},\dots,e^{j2\pi \cdot (M-1)/M}\}
\end{equation*}

\noindent together with the multiplication operation forms a group. Thus each ${\bf V}(2,k)$ corresponds to a unique
${\bf F}_{M\times M}(2,k')$ for some $k'\in\{0,\dots,M-1\}$. Now observe that ${\bf V}(j,k) = {\bf V}(2,k)^{(j-1)}$ and
${\bf F}_{M\times M}(j,k') = {\bf F}_{M\times M}(2,k')^{(j-1)}$. Thus ${\bf F}_{M\times M}(j,k') = {\bf V}_{M\times
M}(j,k)$. In other words each column of $\bf V$ corresponds to a unique column of ${\bf F}_{M\times M}$, which
establishes the result.
\end{proof}

Lemma \ref{lem:V} enables us to write the encoding matrices ${\bf V}_j$ in a revealing form. Define
\begin{align} \label{eqn:Gamma}
{\bf \Gamma}_1 &\triangleq \overline {\bf H}_{31}^{-1}\overline {\bf H}_{32}{\bf T}^2 \\ \label{eqn:Gamma2} {\bf
\Gamma}_2 &\triangleq {\bf T} \\ \label{eqn:Gamma3}
{\bf \Gamma}_3 &\triangleq \overline {\bf H}_{13}^{-1}\overline
{\bf H}_{12}.
\end{align}

\noindent Then
\begin{equation*}
\tilde {\bf x}_j = \left(
                     \begin{array}{c}
                     {\bf 0}_{l_{\max}\times (M-l_{\max})} \;\;{\bf I}_{l_{\max}\times l_{\max}} \\
                     {\bf I}_{M\times M}
                     \end{array}
                   \right) {\bf F}_{M\times M}^*{\bf \Gamma}_j{\bf F}_{M\times M}{\bf \pi}_{l,M}{\bf x}_j.
\end{equation*}

\noindent Examining the above expression reveals that the encoding operation for tx $j$ corresponds to transmitting
consecutive data symbols $l$ time slots apart, but cyclicly wrapped around such that roughly half of all time slots
contain data symbols and no two data symbols share the same time slot. As the operation ${\bf F}_{M\times M}^*{\bf
\Gamma}_j{\bf F}_{M\times M}$ corresponds to delaying the input stream, the entire transmission sequence is just offset
by this amount.

Based on equations (\ref{eqn:Gamma})-(\ref{eqn:Gamma3}) we can define the delay $d_j$ for user $j$'s transmission
sequence as
\begin{align*}
d_1 &\triangleq 2l_{12}-2l_{21}+2l_{23}-l_{32}+l_{31}-2l_{13} \\
d_2 &\triangleq l_{12}-l_{21}+l_{23}-l_{32}+l_{31}-l_{13} \\
d_3 &\triangleq l_{12}-l_{13}.
\end{align*}

\noindent Then we see that transmitter $j$ will send its first symbol $x_j[1]$ in time slot $d_j\mod M$, its second
symbol $x_j[2]$ in time slot $d_j+l \mod M$, its third in time slot $d_j+2l\mod M$, etc... The last symbol will be sent
at time $d_j+(n-1)l\mod M$. More precisely transmitter $j$ sends
\begin{equation*}
\tilde x_j[m] = \left\{
                  \begin{array}{ll}
                    x_j[k], & \hbox{if $m = kl + d_j\mod M$ for some $k$} \\
                    0, & \hbox{otherwise}
                  \end{array}
                \right.
\end{equation*}

\noindent at times $m=0,1,\dots,M-1$.

The manner by which this transmission scheme achieves interference alignment is now simple to understand. Tx 1,
transmits it's first data symbol $x_1[1]$, such that it arrives at rx 3 at the same time as tx 2's first data symbol
$x_2[1]$. Tx 2 transmits $x_2[1]$ such that it arrives at the same time as $x_3[1]$ at rx 1. Tx 3 transmits $x_3[1]$
such that it arrives at rx 2 at the same time as $x_1[2]$, etc... See example \ref{exm:delays} and figure
\ref{fig:3useralign}.

From figure \ref{fig:3useralign} it is clear how decoding should be performed ---receivers merely decode each data
symbol by looking at the time slot in which it was received. But let us reconcile this with the decoding methodology of
\cite{Jafar1}, where the received sequence is passed through a ZF equalizer. This corresponds to projecting the vector
${\bf y}_j$ onto the subspace orthogonal to the interference. At the first receiver
\begin{equation*} {\bf y}_1 = \left(
\begin{array}{cc}
  {\bf H}_{11}{\bf V}_1 & {\bf H}_{12}{\bf U}
  \end{array} \right) \left(
\begin{array}{c}
  {\bf x}_1 \\
  {\bf x}_g
\end{array}\right) + {\bf z}_1
\end{equation*}

\noindent where
\begin{equation*}
{\bf U} = [{\bf w} \; {\bf Tw}\; {\bf T}^2{\bf w} \; \cdots \; {\bf T}^{n}{\bf w}].
\end{equation*}

\noindent and ${\bf x}_g$ represents a combination of interfering symbols from the 2nd and 3rd users. It is
straightforward to see that we can write the space orthogonal to $\bf U$ as
\begin{equation*}
{\bf U}^c = {\bf F}_{M\times M}{\bf \pi}_{l,M}^c
\end{equation*}

\noindent where ${\bf \pi}^c_{l,M}$ is a permutation matrix orthogonal to ${\bf \pi}_{l,M}$ in the sense that ${{\bf
\pi}_{l,M}^c}^*{\bf \pi}_{l,M} = {\bf 0}$. Simply put, ${\bf U}^c$ is a matrix whose columns are those columns of ${\bf
F}_{M\times M}$ that are not present in $\bf V$. Thus the first receiver computes ${\hat {\bf x}}_1 = {\bf G}_1^*{\bf
y}_1$ where
\begin{equation*}
{\bf G}_1^* = ({{\bf U}^c}^*{\bf H}_{12}^{-1}{\bf H}_{11}{\bf V}_1)^{-1}{{\bf U}^c}^*{\bf H}_{12}^{-1}.
\end{equation*}

\noindent If we write out the decoder in detail
\begin{multline*}
{\bf G}_1^*{\bf y}_1 =\\ ({{\bf \pi}_{l,M}^c}^*{\bf F}_{M\times M}^*{\bf H}_{12}^{-1}{\bf H}_{11}{\bf \Gamma}_1{\bf
F}_{M\times M}{\bf \pi}_{l,M})^{-1}{{\bf \pi}_{l,M}^c}^*{\bf F}_{M\times M}^*{\bf H}_{12}^{-1}{\bf
F}_{M\times M}\left(
                     \begin{array}{c}
                     {\bf 0}_{M\times l_{\max}} \;\;{\bf I}_{M\times M}
                     \end{array}
                   \right) \tilde {\bf y}_j
\end{multline*}

\noindent we see that the first few operations correspond to removing the cyclic prefix, delaying the resulting stream
by $l_{12}$ and then selecting the interference free subset of this. The last operation is in general undefined, as the
matrix we invert may not be full rank. However in certain scenarios the matrix equals identity and is then invertible.

The reason for this phenomenon is that the last operation corresponds to recovering the data from the interference free
subspace, which may contain fewer than $n$ dimensions of ${\bf x}_j$. In the case of many independently faded OFDM
sub-channels such as is assumed in \cite{Jafar1}, one could always project the subspace spanned by the data onto the
interference free subspace without losing information, however for single path channels each data dimension is either
orthogonal or overlapping with an interference dimension.

The scenario in which the above decoder is well defined (i.e. data and interference subspaces are orthogonal) is given
by the following condition.
\begin{lemma}
If
\begin{align*}
l_{11} &= d_1+nl\mod M \\
l_{22} &= d_2+nl\mod M \\
l_{33} &= d_3+nl\mod M
\end{align*}

\noindent then the interference subspace is orthogonal to the data subspace.
\end{lemma}

\noindent Thus if each of the direct delays takes on a single, specific value, the signal space will be orthogonal to
the interference space at each of the receivers and we will be able to decode all data symbols.

\begin{figure}
\centering
\includegraphics[width=365pt]{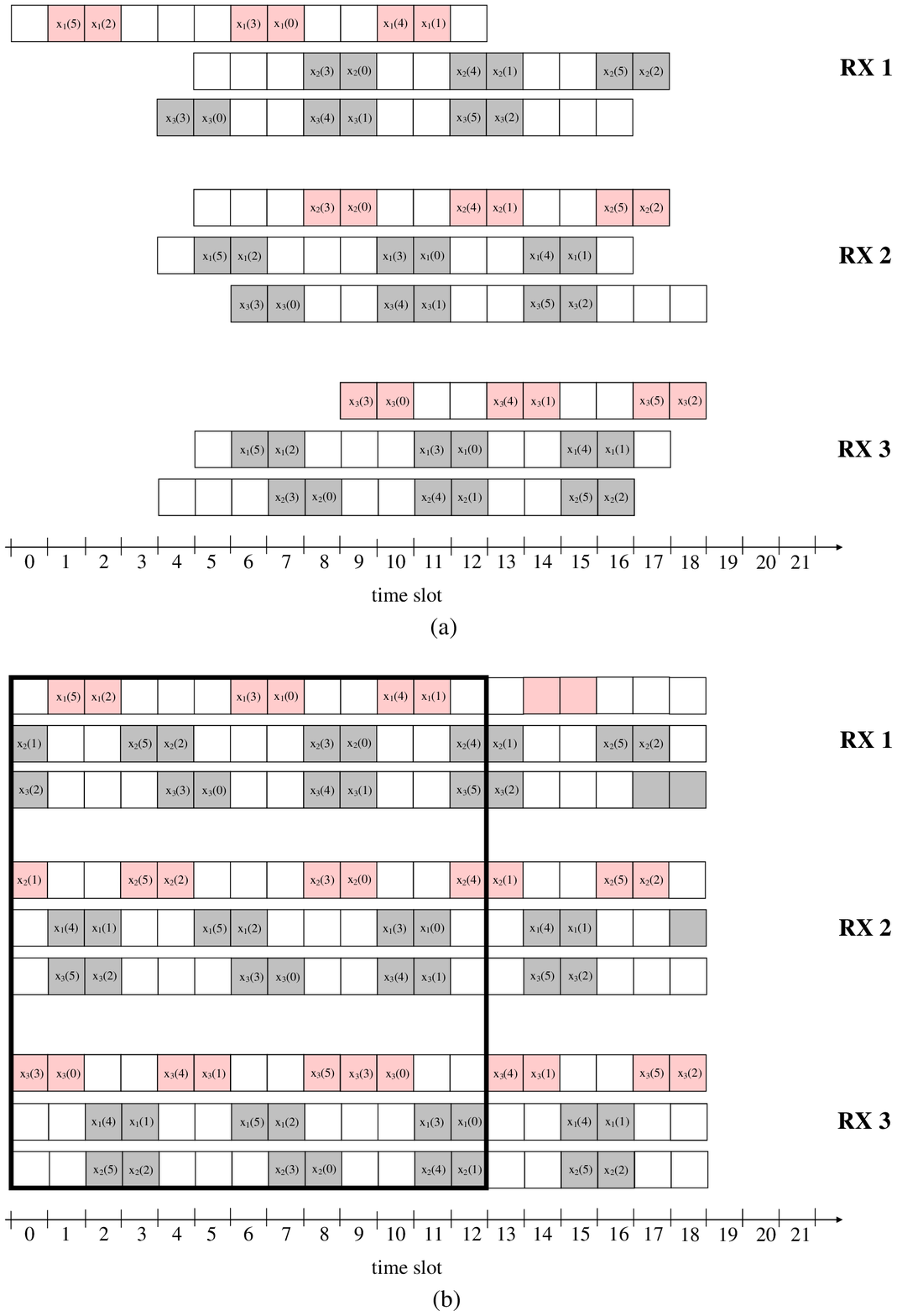}
\caption{Illustration of 3 user interference alignment scheme of [5] in the time domain. See example
\ref{exm:delays} for a description. (a) Received sequences with the cyclic prefix omitted. (b) Received sequences
incorporating the cyclic prefix.} \label{fig:3useralign}
\end{figure}

\begin{example}\label{exm:delays}
Suppose the link delays are $l_{11}=0$, $l_{22}=5$, $l_{33} = 9$, $l_{12} = 5$, $l_{13} = 4$, $l_{21} = 4$, $l_{23} =
6$, $l_{31} = 5$, $l_{32} = 4$. Then $d_1 = 7$, $d_2 = 4$ and $d_3 = 1$. Also $l=5-4+6-4+5-4=4$. Note that for
illustrative purposes the direct delays have been precisely chosen such that the data and interference subspaces are
orthogonal. Choose the data block length $M$ to be the prime 13, and use a cyclic prefix of length $9$. The total
length of cyclic prefix plus data block is 21. Then tx 1 will transmit its 0th data symbol, namely $x_1[0]$, in time
slot $d_1 \mod M = 7 \mod 13 = 7$. It's second data symbol $x_1[1]$ will be transmitted in time slot $l+d_1\mod M = 4 +
7 \mod 13 = 11$. It's third data symbol $x_1[2]$ will be transmitted in time slot $2l+d_1\mod M = 15 \mod 13 = 2$,
etc.. These data symbols will arrive at rx 2 delayed by $l_{21} = 4$ time slots. Thus $x_1[0]$ will appear as
interference at rx 2 during time slot $7+4=11$, $x_1[1]$ will appear as interference during time slot $11+4=15$, etc...
Similarly tx 1's data symbols will arrive at rx 3 delayed by $l_{31} = 5$ time slots. Similarly one can do the same
computation for tx 2's and tx 3's data symbols.

The details are given in figure \ref{fig:3useralign}. In part (a) of the figure the cyclic prefix has been omitted
for illustrative purposes. It is incorporated into the picture in part (b). The red shaded boxes contain data symbols,
the grey shaded boxes contain interference symbols. Notice the interference alignment property manifests itself as an
overlapping of interfering data symbols. The shaded, but unlabeled boxes represent symbols belonging to the next OFDM
block. The black box outlines those time slots that are used for decoding. The 9 time slots prior to these are
discarded when the cyclic prefix is removed. The 9 unboxed time slots to the right of the black outline will also be
discarded, but during the next OFDM block.
\end{example}

\subsection{K-user channels}

In the previous section we demonstrated that for three-user LOS channels, the frequency domain scheme of \cite{Jafar1}
has a simple analog in the time domain that works well when the block length is chosen to be a prime number, and the
direct delays take on particular values. One would imagine that an analogous scheme for $K>3$ users would therefore
also exist and work well. This is not the case. In fact for the LOS interference channel with more than three users,
the alignment scheme of \cite{Jafar1} has various shortcomings which result in it achieving zero degrees of freedom in
total. Our construction (in section \ref{sec:non_vanish_se}) is inspired by the use of a generalized arithmetic
progression in \cite{Jafar1}, but circumvents the schemes shortcomings by:
\begin{enumerate}
  \item Truncating the generalized arithmetic progression appropriately.
  \item Interleaving many replicas of the progression, with random offsets.
  \item Scaling the bandwidth with $K$.
\end{enumerate}

For example, to understand why truncation is necessary, recall the precoding matrices used in for the $K$-user channel
in \cite{Jafar1} are of the form
\begin{equation*}
{\bf V}_j = {\bf S}_j {\bf B}
\end{equation*}

\noindent for $j=2,\dots,K$, where the matrix $\bf B$ is composed of the column vectors in the set
\begin{equation*}
{\cal B} = \left\{ \left( \prod_{i,j\in\{2,3,\dots,K\},i\neq j,(i,j)\neq (2,3)} (\overline {\bf H}_{i1}^{-1}\overline
{\bf H}_{ij}{\bf S}_j)^{\alpha_{ij}}\right){\bf w}: \alpha_{ij} \in \{0,1,\dots,n-1\} \right\},
\end{equation*}

\noindent and ${\bf S}_j = \overline {\bf H}_{1j}^{-1}\overline {\bf H}_{13}\overline {\bf H}_{23}^{-1}\overline {\bf
H}_{21}$. Observing equation (\ref{eqn:ch_matrices}), write the link matrices in the form
\begin{equation*}
\overline {\bf H}_{ij} = {\bf Z}^{l_{ij}},
\end{equation*}

\noindent where
\begin{equation*}
{\bf Z}\triangleq\left(
              \begin{array}{ccccc}
                1 &  &  &  &  \\
                 & e^{-j2\pi\cdot /M} &  &  &  \\
                 &  & e^{-j2\pi\cdot 2/M} &  &  \\
                 &  &  & \ddots &  \\
                 &  &  &  & e^{-j2\pi\cdot (M-1)/M} \\
              \end{array}
            \right).
\end{equation*}

\noindent Then we have
\begin{equation*}
{\cal B} = \left\{ {\bf Z}^{\sum_{i,j\in\{2,3,\dots,K\},i\neq j,(i,j)\neq (2,3)} \alpha_{ij}\tilde l_{ij}}{\bf w}:
\forall \alpha_{ij} \in \{0,1,\dots,n-1\} \right\}.
\end{equation*}

\noindent where $l_{ij} = -l_{ij}+l_{ij}-l_{1j}+l_{13}-l_{23}+l_{21}$. For sufficiently large $n$ we will be able to
find many pairs $(\{\alpha_{ij}\},\{\alpha_{ij}'\})$ such that
\begin{equation*}
\sum_{i,j\in\{2,3,\dots,K\},i\neq j,(i,j)\neq (2,3)} \alpha_{ij}\tilde l_{ij} = \sum_{i,j\in\{2,3,\dots,K\},i\neq
j,(i,j)\neq (2,3)} \alpha_{ij}'\tilde l_{ij}.
\end{equation*}

\noindent Thus the precoding matrices {\it will loose rank} as many of their columns will be repeats of previous ones.
This phenomenon of repeated elements is common in generalized arithmetic progressions over integer fields. See for
example \cite{Tao}. How much rank will be lost? Observe that the largest exponent of ${\bf Z}$ in ${\cal B}$ will be no
greater than $(n-1)((K-1)(K-2)-1)\max l_{ij}$. But there are $n^{(K-1)(K-2)-1}$ columns in $\bf B$. Thus as
$n\rightarrow \infty$ the rank of $\bf B$ will scale only like $O(n)$ due to repeated columns, whilst the dimension of
the space scales like roughly $O(n^{(K-1)(K-2)-1})$. Hence the total degrees of freedom goes to zero unless the
progression is truncated.

\subsection{Bandwidth Scaling Revisited}

The final issue we address in terms of reconciling time and frequency domain interpretations, is that of bandwidth
scaling. We now demonstrate that bandwidth scaling is required in the scheme of \cite{Jafar1} when the physical channel
model is brought into the picture.

\begin{theorem}\label{thm:bsr}
In a multipath fading channel with $L$ taps, if the bandwidth satisfies
\begin{equation*}
\lim_{K\rightarrow\infty} \frac{\log W}{\log ((K-1)(K-2)-1)^{(K-1)(K-2)-3}} = 0,
\end{equation*}

\noindent then the total degrees of freedom achieved by the $K$-user interference alignment scheme of \cite{Jafar1}
goes to zero as $K\rightarrow \infty$.
\end{theorem}

\noindent This means that the bandwidth must scale at least as fast as $O(((K-1)(K-2)-1)^{(K-1)(K-2)-3})$ which is
roughly the same scaling that is required in theorem \ref{thm:mr}, namely $O(K^{2K^2})$.

\begin{proof}
\noindent For a general multipath channel with $L$ taps, the link matrices are of the form
\begin{equation*}
{\bf H}_{ij} = \sum_{l=0}^{L-1} a_{ij,l} {\bf Z}^{l}.
\end{equation*}

\noindent Using the commutativity of the diagonal ${\bf H}_{ij}$ matrices we can write
\begin{equation*}
{\bf V}_j = {\bf S}_j \left(\prod_{i,j\in\{2,3,\dots,K\},i\neq j,(i,j)\neq (2,3)} {\bf H}_{i1}{\bf H}_{1j} {\bf H}_{23}
\right)^{-n} {\bf C},
\end{equation*}

\noindent where the matrix $\bf C$ is composed of the column vectors in the set
\begin{multline*}
{\cal C} = \Bigg\{ \left( \prod_{i,j\in\{2,3,\dots,K\},i\neq j,(i,j)\neq (2,3)} ({\bf H}_{i1}{\bf H}_{1j}{\bf
H}_{23})^{n-\alpha_{ij}}({\bf H}_{ij}{\bf H}_{13}{\bf H}_{21})^{\alpha_{ij}}\right){\bf w}: \\
\alpha_{ij} \in \{0,1,\dots,n-1\} \Bigg\}.
\end{multline*}

\noindent In \cite{Jafar1} the minimum scaling of $n$ with $K$ required is
\begin{equation*}
\lim_{K\rightarrow\infty} \frac{\log n}{\log (K-1)(K-2)-1} > 0.
\end{equation*}

Each of the ${\bf H}_{ij}$ matrices is a polynomial of degree at most $L-1$ in the matrix ${\bf Z}$. Thus each column
of $\bf C$ is a polynomial of degree at most $6n(L-1)((K-1)(K-2)-1)$ in the matrix ${\bf Z}$. This means the maximum
rank of $\bf C$ is $6n(L-1)((K-1)(K-2)-1)+1$, as any $c$ polynomials of degree $\le d$ that are in general position, are
linearly dependent for $c>d+1$. The total number of rows in $\bf C$ however, is at least $n^{(K-1)(K-2)-1}$. Thus the total degrees
of freedom is no more than
\begin{equation*}
\frac{6n(L-1)((K-1)(K-2)-1)}{n^{(K-1)(K-2)-1}}.
\end{equation*}

\noindent which goes to zero as $K\rightarrow\infty$, unless $L$ (and hence $W$) scales like
\begin{equation*}
\lim_{K\rightarrow\infty} \frac{\log L}{\log ((K-1)(K-2)-1)^{(K-1)(K-2)-3}} > 0.
\end{equation*}
\end{proof}

\section{Discussion and Conclusion}\label{sec:conc}

We demonstrated in section \ref{sec:BS} that if the bandwidth scales sub-linearly in $K$, then the independence rate of
the interference graph goes to zero as $K\rightarrow\infty$, and in section \ref{sec:non_vanish_se}, that if the
bandwidth scales like $O((2K(K-1))^{K(K-1)})$, then the independence rate scales arbitrarily close to $O(K)$. This
brings us to an interesting open question. How does the independence rate scale in the intermediate regime where
$O(K/\log K) < W < O((2K(K-1))^{K(K-1)})$? What about the $D$-path channel?

There are several variations of the LOS interference channel for which the interference graph techniques discussed in
this work are applicable, and for which further study is warranted. These were alluded to earlier. For instance, the partial connected interference channel is a more accurate model of an extended wireless adhoc
network. It is possible that for such channels, a bandwidth scaling much less than $O((2K(K-1))^{K(K-1)})$ is
sufficient in order for the independence rate to scale arbitrarily close $O(K)$. It is not clear how one would approach
the problem of showing this if it were true, or disproving it otherwise. An interference channel with one dominant path
per link and several sub-dominant ones, is also an interesting candidate for investigation. This $D$-path channel is
commonly encountered in practice. It is possible that an optimization problem similar to
\ref{eqn:op_problem}, but allowing for signal and interference to overlap, can be formulated for this scenario. It would be interesting to study whether time-domain based
interference alignment can provide gains in this case. Presumably, for fixed $W$, as the dominance of one physical path over the others
diminishes, so too will the performance improvement.

Lastly we discuss some interesting fringe benefits associated with the communication schemes presented in this work.
Whereas interference alignment in the frequency domain requires coding over very long blocks, which results in
substantial delay due to the necessity of buffering data symbols at the encoder and received symbols at the decoder, no
such delay is required for the time-indexed interference graph techniques detailed above. Data symbols are transmitted
as soon as an appropriate time slot is reached, and detected when received. The only delay incurred is that stemming
from the use of an error correction code. In the same respect, the encoding and decoding complexity are greatly
reduced. Thus delay and complexity issues are non-existent.

Subspace conditioning issues are also non-existent. Interference alignment in the frequency domain, although performing well at
very high-$\PSD$, suffers at moderate $\PSD$ if the data and interference subspaces are close to one another. Time
domain interference alignment techniques are free from this problem as the data and interferences subspaces are
orthogonal by design.

\section{Acknowledgements}

The authors would like to thank Professor Satish Rao for contributing the converse argument of section \ref{sec:BS}, and Professor Tom Luo for the argument behind the proof of theorem \ref{thm:bsr}.

\section{Appendix}

\subsection{Proof of Theorem \ref{thm:C}}

First suppose $l\neq 0$. From the interference graph form $l$ infinite {\it chain graphs} ${\cal G}'_0,\dots,{\cal
G}'_{l-1}$. These graphs will be functions of $K$ and $\{l_{ij}\}_{i\neq j}$ but for notational brevity we omit this
notation. The $i$th chain graph ${\cal G}'_i$ has vertex and edges sets
\begin{align*}
{\cal V}'_i &= \{v'_i(0),v'_i(1),\dots\} \\
{\cal E}'_i &= \{e'_i(0),e'_i(1),\dots\}.
\end{align*}

\noindent This is an undirected graph with edge $e'_i(j)$ joining vertices $v'_i(j)$ and $v'_i(j+1)$. We now associate
the chain graphs with the interference graph. Let
\begin{align*}
v'_i(6k) &= v_3(kl+i) \\
v'_i(6k+1) &= v_1(kl+i+l'_{31}) \\
v'_i(6k+2) &= v_3(kl+i+l'_{31}+l'_{13}) \\
v'_i(6k+3) &= v_2(kl+i+l'_{31}+l'_{13}+l'_{32}) \\
v'_i(6k+4) &= v_1(kl+i+l'_{31}+l'_{13}+l'_{32}+l'_{21}) \\
v'_i(6k+5) &= v_2(kl+i+l'_{31}+l'_{13}+l'_{32}+l'_{21}+l'_{12})
\end{align*}

\begin{figure}
\centering
\includegraphics[width=430pt]{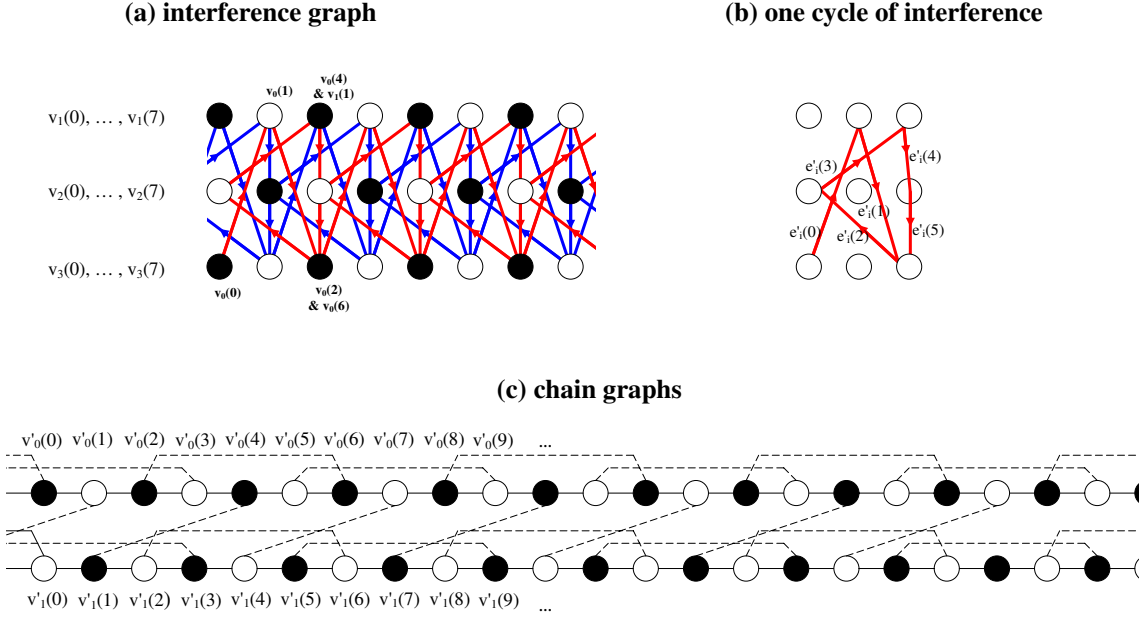}
\caption{(a) A segment of the interference graph. Each row of vertices corresponds to the transmission opportunities
for each of the users. The shaded vertices correspond to a feasible transmit pattern achieving an independence rate of
$3/2$. A few of the vertices are labeled with their equivalent vertices in the chain graphs. In this example the normalized cross delays are $l'_{21} =
0$, $l'_{31}=1$, $l'_{12}=2$, $l'_{32}=0$, $l'_{13}=1$ and $l'_{23}=-2$. Thus $l=0+1+2+0+1-2=2$ meaning that each cycle
of interference moves two time slots to the right as illustrated in (b), which shows a single cycle from the
interference graph, containing all six directed edges. (c) The corresponding chain graphs. As $l=2$, there are two
chains. A pair of twin vertices connected by a dashed line in the chain graphs, correspond to the same vertex in the
original interference graph. }\label{fig:chain_graphs}
\end{figure}

\noindent where the $v_k(t)$ are the vertices in the original interference graph. See figure \ref{fig:chain_graphs} for
an illustration. Note that the mapping from the interference graph to the $l$ chain graphs is not one-one. In
particular, each vertex in the interference graph is associated with a {\it pair} of vertices in the set of chain
graphs. Paired vertices are called {\it twins} as they correspond to the same vertex from the original interference
graph. We now think of a feasible transmit pattern as a collection of vertices from the set of chain graphs. However,
note that if a feasible transmit pattern results in a particular vertex (from the set of chain graphs) being included
in the independent set, its twin will also be included. Likewise if a feasible transmit pattern results in a particular
vertex being excluded, its twin will also be excluded. The key to characterizing the set of channels for which an
interference rate of 3/2 is achievable lies in understanding which pairings are favorable, and which are not.

The pairings can be succinctly described by the following three equations
\begin{align} \label{eqn:IO_conditions1}
v'_i(6k) &= v'_{i-l_3\modd l}(6(k+\lfloor l_3/l \rfloor)+2) \\
\label{eqn:IO_conditions2} v'_i(6k+1) &= v'_{i-l_1\modd l}(6(k+\lfloor l_1/l \rfloor)+4) \\
\label{eqn:IO_conditions3} v'_i(6k+3) &= v'_{i-l_2\modd l}(6(k+\lfloor l_2/l \rfloor)+5)
\end{align}

\noindent for $i=0,\dots,l-1$ and $k=0,1,\dots$.

In order to achieve an independence rate of 3/2, half of all vertices must be included in the transmit pattern. Denote
the transmit pattern by ${\cal T}$. Because in each chain, all neighboring vertices are connected by an edge, this is
only possible if in each chain, every second node is included in the transmit pattern. For each chain there are two
ways of doing this, either $v'_i(2k)\in{\cal T}$ for all $k$, or $v'_i(2k+1)\in{\cal T}$ for all $k$. Let $c_i$ denote
the phase of the $i$th chain. If the former condition holds, we say the chain is {\it in phase} and write $c_i = I$. If
the latter holds we say the chain is {\it out of phase} and write $c_i=O$. In the entire graph there are only $2^l$
combinations we need to examine, corresponding to all possible inphase/out of phase assignments for the $l$ chains. A
feasible transmit pattern achieving independence rate $3/2$ exists if and only if each chain admits an I or O
assignment and the assignment of I's and O's to the $l$ chains does not violate conditions
(\ref{eqn:IO_conditions1})-(\ref{eqn:IO_conditions3}). Thus we wish to characterize those channels for which such an
I/O assignment can be found.

At this point we consider an example. Suppose $l$ divides $l_1$. We claim an independence rate of 3/2 is not
achievable. To see this argue by contradiction. Assume that $c_0 = I$. As $l$ divides $l_1$, we have $l_1\modd l=0$ and
condition (\ref{eqn:IO_conditions2}) pairs vertex $v_0(1)$ with vertex $v_0(6k'+4)$ for some integer $k'$. But these
vertices lie an odd distance apart on the same chain, so working backwards we see that we must have $c_0=O$, a
contraction. So suppose instead that $c_0 = 0$. Using the same logic as before we arrive at $c_0 = I$, again a
contradiction. Thus the independence rate is less than 3/2.

From this example we see that condition (\ref{eqn:IO_conditions1}) tells us if $c_0=I$ then we must also have
$c_{-l_3\modd l}=I$, as vertices $v'_0(0)$ and $v'_{-l_3\modd l}(6\lfloor l_3/l \rfloor+2)$ are an even distance apart.
Continuing this logic we see that we must also have $c_{-l_1\modd l}=O$, and $c_{-l_2\modd l}=I$. We can also conclude
something else, as $c_{-l_1\modd l}=O$, we must have $c_{-2l_1\modd l}=I$ by condition (\ref{eqn:IO_conditions2}).
Continuing further, we must have $c_{-2l_1-l_2\modd l}=I$ by condition (\ref{eqn:IO_conditions3}) and so on.

By this point it should be clear that conditions (\ref{eqn:IO_conditions1})-(\ref{eqn:IO_conditions3}) are satisfied if
and only if for all integers $k_1,k_2,k_3$,
\begin{equation} \label{eqn:c_cond}
c_{2k_1l_1+k_2l_2+k_3l_3\modd l} \neq c_{(2k_1+1)l_1+k_2l_2+k_3l_3\modd l}
\end{equation}

\noindent Let ${\cal P}(l)$ denote the group consisting of integers $\{0,\dots,l-1\}$ together with the addition modulo
$l$ operation. Consider the set of chains $c_{2k_1l_1+k_2l_2+k_3l_3\modd l}$ for all integers $k_1,k_2,k_3$. This set
forms a subgroup of ${\cal P}(l)$ with generator $\gcd(2l_1,l_2,l_3)$. We denote this subgroup by ${\cal
P}_{\gcd(2l_1,l_2,l_3)}(l)$. It has $\gcd(2l_1,l_2,l_3,l)-1$ cosets other than itself. The set of chains
$c_{(2k_1+1)l_1+k_2l_2+k_3l_3\modd l}$ for all integers $k_1,k_2,k_3$, forms coset number $l_1 \modd
\gcd(2l_1,l_2,l_3,l)$. But as
\begin{equation*}
2l_1 \modd \gcd(2l_1,l_2,l_3,l)=0,
\end{equation*}

\noindent either $l_1 \modd \gcd(2l_1,l_2,l_3,l)$ equals $\gcd(l_1,l_2/2,l_3/2,l/2)$, or 0. If it is zero, condition
(\ref{eqn:c_cond}) above is violated. This occurs if and only if $l_1$ is a multiple of $\gcd(2l_1,l_2,l_3,l)$.
Alternatively if it equals $\gcd(l_1,l_2/2,l_3/2,l/2)$ then we can choose the phases of half the cosets, namely cosets
\begin{equation*}
0,1,..., \gcd(l_1,l_2/2,l_3/2,l/2)-1
\end{equation*}

\noindent arbitrarily, and still satisfy (\ref{eqn:c_cond}). For this reason we refer to the chains
\begin{equation*}
c_0,c_1,..., c_{\gcd(l_1,l_2/2,l_3/2,l/2)-1}
\end{equation*}

\noindent as seed chains. This means that there are $2^{\gcd(l_1,l_2/2,l_3/2,l/2)}$ possible solutions that achieve
independence rate $3/2$. So what does it mean for $l_1$ to not be a multiple of $\gcd(2l_1,l_2,l_3,l)$? It means that
\begin{align*}
\gcd(l_1,2l_1,l_2,l_3,l) \neq \gcd(2l_1,l_2,l_3,l).
\end{align*}

\noindent In other words
\begin{align*}
\gcd(l_1,l_2,l_3,l) \neq \gcd(2l_1,l_2,l_3,l).
\end{align*}

\noindent It is shown in lemma \ref{lem:alpha} that the above inequality is equivalent to having $\gamma_1<\gamma_2$,
$\gamma_1<\gamma_3$ and $\gamma_1<\gamma$. This establishes theorem \ref{thm:C} for $l\neq 0$.

Now suppose $l=0$. This proof is a slight modification of the previous. From the interference graph from an infinite number of {\it cycle graphs} ${\cal G}_0',{\cal G}_1',\dots$. The $i$th cycle graph ${\cal G}_i'$ has vertex and edge sets
\begin{align*}
{\cal V}_i' &= \{v_i'(0),v_i'(1),v_i'(2),v_i'(3),v_i'(4),v_i'(5)\} \\
{\cal E}_i' &= \{e_i'(0),e_i'(1),e_i'(2),e_i'(3),e_i'(4),e_i'(5)\},
\end{align*}

\noindent where edge $e_i'(j \mod 6)$ joins vertices $v_i(j \mod 6)$ and $v_i(j+1 \mod 6)$, for $j=0,1,2,3,4,5$. Notice that for $l\neq 0$ we created a finite number of chain graphs, each with an infinite number of vertices, whereas for $l=0$ we create an infinite number of cycle graphs, each with a finite number of vertices. We now associate the cycle graphs with the interference graph. Let
\begin{align*}
v'_i(0) &= v_3(i) \\
v'_i(1) &= v_1(i+l'_{31}) \\
v'_i(2) &= v_3(i+l'_{31}+l'_{13}) \\
v'_i(3) &= v_2(i+l'_{31}+l'_{13}+l'_{32}) \\
v'_i(4) &= v_1(i+l'_{31}+l'_{13}+l'_{32}+l'_{21}) \\
v'_i(5) &= v_2(i+l'_{31}+l'_{13}+l'_{32}+l'_{21}+l'_{12})
\end{align*}

\noindent where $v_k(t)$ are the vertices in the original interference graph. As before, the mapping from interference graph to the indefinite number of cycle graphs is not one-one. Each vertex in the interference graph is associated with a pair of vertices in the set of cycle graphs. The pairings are described by the following three equations
\begin{align*}
v'_i(0) &= v'_{i-l_3}(2) \\
v'_i(1) &= v'_{i-l_1}(4) \\
v'_i(3) &= v'_{i-l_2}(5)
\end{align*}

We want to assign each cycle graph $i$ a phase, either $c_i=I$ or $c_i=O$ and we need to find necessary and sufficient conditions for a feasible assignment. Similarly to equation \ref{eqn:c_cond}, the condition we are after is
\begin{equation*}
c_{2k_1l_1+k_2l_2+k_3l_3} \neq c_{(2k_1+1)l_1+k_2l_2+k_3l_3}
\end{equation*}

\noindent for all integers $k_1,k_2,k_3$. By this point it should be clear, based on the proof for the $l\neq 0$ case, that the above condition is equivalent to
\begin{equation*}
\gcd(l_1,l_2,l_3) \neq \gcd(2l_1,l_2,l_3).
\end{equation*}

\noindent This inequality is equivalent to having $\gamma_1<\gamma_2$ and $\gamma_1<\gamma_3$. Thus we see that the conditions for $l=0$ case are the same as the $l\neq 0$ case if we set $\gamma =\infty$. This establishes the result in general.

\begin{lemma} \label{lem:alpha}
$\gcd(l_1,l_2,l_3,l) \neq \gcd(2l_1,l_2,l_3,l)$ if and only if $\gamma_1<\gamma_2$, $\gamma_1<\gamma_3$ and $\gamma_1<\gamma$.
\end{lemma}

\begin{proof}
Suppose
\begin{align*}
\gcd(l_1,l_2,l_3,l) &\neq \gcd(2l_1,l_2,l_3,l) \\
\Rightarrow \gcd(2^{\gamma_1}\beta_1,2^{\gamma_2}\beta_2,2^{\gamma_3}\beta_3,2^{\gamma}\beta) &\neq \gcd(2^{\gamma_1+1}\beta_1,2^{\gamma_2}\beta_2,2^{\gamma_3}\beta_3,2^{\gamma}\beta) \\
\Rightarrow \gcd(2^{\gamma_1},2^{\gamma_2},2^{\gamma_3},2^{\gamma})\gcd(\beta_1,\beta_2,\beta_3,\beta) &\neq \gcd(2^{\gamma_1+1},2^{\gamma_2},2^{\gamma_3},2^{\gamma})\gcd(\beta_1,\beta_2,\beta_3,\beta) \\
\Rightarrow \min(\gamma_1,\gamma_2,\gamma_3,\gamma) &\neq \min(\gamma_1+1,\gamma_2,\gamma_3,\gamma).
\end{align*}

\noindent This can only hold if $\gamma_1<\gamma_2$, $\gamma_1<\gamma_3$ and $\gamma_1<\gamma$. The proof in the opposite direction is identical.
\end{proof}

\end{document}